\documentclass[12pt]{iopart}

\pdfoutput=1

\usepackage{iopams}
\usepackage{float}
\usepackage{amsfonts}
\usepackage{amssymb}
\usepackage{bm}
\usepackage{color}
\usepackage{subfigure}
\usepackage{graphicx}
\usepackage{amsthm}
\usepackage{enumerate}
\usepackage{mathrsfs}
\usepackage[utf8x]{inputenc}
\usepackage{lmodern,textcomp}
\usepackage{amsthm}

\newtheorem{corollary}{Corollary}

\begin{document}

\title{Integral Fluctuation Relations  for Entropy Production 
 at Stopping Times}  

\author{Izaak Neri$^{1}$, Edgar Rold\'{a}n$^{2}$, Simone Pigolotti$^{3}$,\\ and Frank J\"{u}licher$^{4}$}

\address{$^1$ Department of Mathematics, King's College London, Strand, London, WC2R~2LS, UK}
\address{$^2$ ICTP - The Abdus Salam International Centre for Theoretical Physics, Strada Costiera 11, 34151 Trieste, Italy}
\address{$^3$ Biological Complexity Unit, Okinawa Institute for Science and Technology and Graduate University, Onna, Okinawa 904-0495, Japan}
\address{$^4$ Max Planck Institute for the Physics of Complex Systems, N\"{o}thnitzerstrasse 38, 01187 Dresden, Germany}

\begin{abstract}
A stopping time $T$ is the first time when a trajectory of a stochastic process satisfies a specific criterion. In this paper, we use martingale theory to derive the integral fluctuation relation $\langle e^{-S_{\rm tot}(T)}\rangle=1$ for the stochastic entropy production $S_{\rm tot}$ in a  stationary physical system  at stochastic stopping times $T$. This fluctuation relation implies the law $\langle S_{\rm tot}(T)\rangle\geq 0$, which states that it is not possible to reduce entropy on average, even by stopping a stochastic process at a   stopping time, and which we call the second law of thermodynamics at stopping times.    This  law  implies bounds on the average amount of heat  and work a system can extract from its environment when stopped at a random time.      Furthermore,  the integral fluctuation relation implies that certain fluctuations of entropy production are universal or are bounded by universal functions.     These universal properties descend from the integral fluctuation relation by selecting appropriate stopping times: for example,  when  $T$  is a first-passage time for entropy production, then we obtain a bound on the  statistics of negative records of entropy production.      We
illustrate these results on simple models of nonequilibrium systems
described by Langevin equations and  reveal two interesting phenomena.   First, 
we demonstrate that isothermal mesoscopic systems can extract
on average heat from their  environment  when  stopped  at  a  cleverly  chosen  moment and  the second law at stopping times provides a bound on the average extracted heat.  Second, we demonstrate that  the average efficiency at stopping times of an autonomous stochastic heat engines, such as  Feymann's ratchet,  can  be larger than the Carnot efficiency and the  second  law  of thermodynamics at stopping times provides a bound on the average efficiency at stopping times.  
  \end{abstract}

\maketitle 

\section{Introduction and statement of the main results}

  Stochastic thermodynamics  is a thermodynamics theory for   the slow degrees of freedom $\vec{X}(t)$ of  a mesoscopic system that is weakly coupled to an environment in equilibrium~\cite{maes2003origin, sekimoto2010stochastic, jarzynski2011equalities, seifert2012stochastic,van2015ensemble}.       Examples of systems to which stochastic thermodynamics applies are  molecular motors \cite{julicher1997modeling, reimann2002brownian}, biopolymers \cite{bustamante2005nonequilibrium}, self-propelled Brownian particles  \cite{bechinger2016active},   micro-manipulation experiments on colloidal particles  \cite{gavrilov2014real,martinez2017colloidal,martinez2016brownian}, and electronic circuits~\cite{ciliberto2017experiments,pekola2015towards, pekola2018thermodynamics}.  
  
     The  stochastic entropy production $S_{\rm tot}(t)$ is a key variable in stochastic thermodynamics.  It is defined as the sum of the entropy change of the environment $S_{\rm env}(t)$ and a    system entropy change $\Delta S_{\rm sys}(t)$~\cite{seifert2005entropy}. In stochastic thermodynamics   entropy production is a functional of the trajectories  of the  slow degrees of freedom in the system.   
 If time is discrete  and  $\vec{X}(t)$ is a variable of even parity with respect to time reversal, then the entropy production $S_{\rm tot}(t)$ associated with  a trajectory of a nonequilibrium stationary process $\vec{X}(t)$ is  the logarithm of the ratio between the stationary probability density  of that trajectory $p(\vec{X}(1), \vec{X}(2), \ldots, \vec{X}(t))$  and the probability density  of the same trajectory but in time-reversed order~\cite{seifert2012stochastic, lebowitz1999gallavotti, crooks1998nonequilibrium, gaspard2004time},
\begin{eqnarray}
S_{\rm tot}(t) =\log  \frac{p(\vec{X}(1), \vec{X}(2), \ldots, \vec{X}(t))}{p(\vec{X}(t), \vec{X}(t-1), \ldots, \vec{X}(1))},  \label{eq:Markov0}
\end{eqnarray} 
where $\log$ denotes natural logarithm. Here and throughout the paper we use dimensionless units for which Boltzmann's constant $k_{\rm B}=1$. Equation~(\ref{eq:Markov0}) is a particular case of the general expression of stochastic entropy production in terms of probability measures that we will discuss below in Eq.~(\ref{eq:SGeneral}).  The functional $S_{\rm tot}(t)$ is exactly equal to zero at all times for systems in equilibrium. For nonequilibrium systems, entropy production fluctuates with expected value larger than zero, $\langle S_{\rm tot}(t)\rangle\geq 0$. 
 
  An interesting consequence of  definition~(\ref{eq:Markov0}) is that the exponential of the negative entropy production $e^{-S_{\rm tot}(t)}$ is a {\it martingale} associated with  the process $\vec{X}(t)$ \cite{chetrite2011two, neri2017statistics, pigolotti2017generic}. 
Historically the concept of martingales has been introduced to understand fundamental questions in betting strategies and gambling \cite{maitra2012discrete}. 
Martingale theory \cite{doob1953stochastic, williams1991probability, liptser2013statistics}, and in particular Doob's optional stopping theorem, provides an elegant resolution to the question whether  it is possible to make profit in a fair game of chance by  leaving the game  at a cleverly chosen moment.    We can distinguish unfair games of chances, where the expected values of the gambler's fortune decreases (or increases) with time, from fair ones, where such expected values are constant in time on average. In probability theory, these categories correspond to supermartingales (or submartingales) and martingales, respectively.
 In a nutshell,  the optional stopping theorem for martingales  states that a gambler  cannot make profit on average in a fair  game of chance by quitting the game at an intelligent chosen moment. 
  The optional stopping theorem  holds as long as the total amount of available money is finite.  A gambler with access to an infinite budget of money could indeed devise a betting strategy that makes profit out of a fair game; the St.~Petersburg game provides an example of a such a strategy, see \cite{risk1954exposition} and chapter 6 of \cite{bernstein1996against}.      Nowadays martingales have  various applications, for example, they model stock prices in efficient capital markets~\cite{leroy1989efficient}.

 In this paper, we study universal properties of entropy production in nonequilibrium stationary states using martingale theory and Doob's optional stopping theorem.   In an analogy with gambling, the  negative entropy production  $-S_{\rm tot}(t)$ of a stationary  process  $\vec{X}(t)$ is equivalent to  a gambler's fortune  in an unfair  game of chance $\vec{X}(t)$ and the exponentiated negative entropy production $e^{-S_{\rm tot}(t)}$ is a martingale associated with the gambler's fortune.   In stochastic thermodynamics,  Doob's optional stopping theorem  implies
 \begin{eqnarray}
\langle e^{-S_{\rm tot}(T)}\rangle =1,  \label{eq:strong}
\end{eqnarray} 
where the expected value $\langle \cdot \rangle$ is over many realizations of the physical process~$\vec{X}(t)$, and where $T$ is a stopping time.  A stopping time $T$  is the first time when a trajectory of  $\vec{X}$ satisfies a specific criterium;  it is thus a stochastic time. This criterium must obey causality and cannot depend on the future.   
  The relation (\ref{eq:strong}) holds under the condition that either $T$ acts in a finite-time window, i.e., $T\in [0,\tau]$ with $\tau$ a positive number, or that $S_{\rm tot}(t)$ is bounded for all times $t\in [0,T]$.      We call (\ref{eq:strong}) the {\it integral fluctuation relation for entropy production at stopping times} because it is an integral relation, $\langle e^{-S_{\rm tot}(T)}\rangle = \int {\rm d}\mathbb{P}\: e^{-S_{\rm tot}(T)} = 1$, that characterises the fluctuations of entropy production.    Here, $\mathbb{P}$ is the probability measure associated with~$\vec{X}(t)$.  
  
  The fluctuation relation at stopping times (\ref{eq:strong}) can be extended into a  fluctuation relation  conditioned on  trajectories $\left\{\vec{X}(t)\right\}_{t\in[0,T']}$ of random duration $[0,T']$, namely, 
 \begin{eqnarray}
\Big\langle e^{-S_{\rm tot}(T)}|\vec{X}^{T'}_0 \Big\rangle = e^{-S_{\rm tot}(T')},  \label{eq:strongCon}
\end{eqnarray} 
with $T'$ a stopping time for which $T'\leq T$, and where  $\vec{X}^s_0 = \left\{\vec{X}(t')\right\}_{t'\in[0,s]}$ denotes a trajectory in a finite-time window.   Notice that for $T'=0$ we obtain $\langle e^{-S_{\rm tot}(T)}| \vec{X}(0)\rangle  = e^{-S_{\rm tot}(0)} = 1$, since in our definitions $S_{\rm tot}(0) = 0$.  The fluctuation relation (\ref{eq:strongCon}) implies thus (\ref{eq:strong}).

There are two important implications of the integral fluctuation relations (\ref{eq:strong}) and~(\ref{eq:strongCon}).  First, it holds that 
\begin{eqnarray}
\langle  S_{\rm tot}(T)\rangle \geq 0 \label{eq:strong2}, 
\end{eqnarray} 
or in other words,  it is not possible to reduce the average   entropy production by stopping a stochastic process $\vec{X}(t)$ at a cleverly chosen moment $T$ that can be different for each realisation.    The relation (\ref{eq:strong2}) is reminiscent of the  second law of thermodynamics, and therefore we call it   the {\it second law of thermodynamics at stopping times}. 
 A second  implication  of the integral fluctuation relations (\ref{eq:strong}) and~(\ref{eq:strongCon}) is that certain fluctuation properties of entropy production are {\it universal}.     In what follows, we discuss in more detail these two consequences of the  integral fluctuation relations. 

 We first discuss the second law  at stopping times Eq.~(\ref{eq:strong2}).   Remarkably, this second law holds for {\em any} stopping time $T$ defined by a (arbitrary) stopping criterium that obeys causality and does not use information about the future of the physical process;  first-passage times are canonical examples of stopping times.   Interestingly, the inequality (\ref{eq:strong2}) bounds the average amount of heat and work a system can perform on  or extract from its surroundings at a stopping time  $T$.     For isothermal systems, Eq. (\ref{eq:strong2}) implies
\begin{eqnarray}
\langle Q (T )\rangle  \leq \mathsf{T}_{\rm env}  \langle  \Delta S_{\rm sys}(T )\rangle \label{eq:QT}.   
\end{eqnarray}    
Relation  (\ref{eq:QT}) states that the average amount of heat $\langle Q (T )\rangle$ a system can on average extract at a stopping time $T$  from a thermal reservoir 
at temperature $\mathsf{T}_{\rm env}$  is smaller or equal than the  average system entropy difference    $\langle  \Delta S_{\rm sys}(T )\rangle $ between the  initial state and the final state at the stopping time.     Similar considerations allow us to derive bounds on the average amount of work that a   stationary heat engine, e.g. Feynman's ratchet,   can extract from its surrounding when stopped at a cleverly chosen moment.  
Consider a system in contact with two thermal reservoirs at   temperatures $\mathsf{T}_{\rm h}$ and $\mathsf{T}_{\rm c}$ with $\mathsf{T}_{\rm h} \geq \mathsf{T}_{\rm c}$.   We define the {\em stopping-time efficiency} $\eta_{T}$   associated with the stopping time $T$ as 
\begin{eqnarray}
\eta_{T} := -\frac{\langle W(T)\rangle}{\langle Q_{\rm h}(T)\rangle},
\label{eq:6}
\end{eqnarray}
where $\langle W(T)\rangle$ is the average work exerted on the system in the time interval $[0,T]$, and $\langle Q_{\rm h}(T)\rangle$ is the average heat absorbed by the system from the hot reservoir  within the same time interval. 
If  $ \langle Q_{\rm h}(T)\rangle>0$, then
the second law at stopping times~(\ref{eq:strong2}) implies that 
  \begin{eqnarray} 
\eta_{T} \    \leq    \eta_{\rm C} -  \frac{\langle \Delta \mathcal{F}_{\rm c}(T)\rangle}{\langle Q_{\rm h}(T)\rangle }, \label{eq:carnotStoppxxxx}
  \end{eqnarray} 
  where $\eta_{\rm C} = 1 -(\mathsf{T}_{\rm c}/\mathsf{T}_{\rm h})$ is the Carnot efficiency,  $\langle \Delta \mathcal{F}_{\rm c}(t)\rangle= \langle \Delta v(t)\rangle -\mathsf{T}_{\rm c}\langle \Delta S_{\rm sys}(t)\rangle$ is the generalised  free energy change of the system at the stopping time $T$, and $\Delta v(t) = v(X(t)) - v(X(0))$ is  change of the  internal  energy of the system.       Note that the second term in the right-hand side of (\ref{eq:carnotStoppxxxx}) can be positive, and thus efficiencies at stopping times of stationary  heat engines  can be greater than the Carnot efficiency. This is because using a
stopping time the system is, in general, no longer  cyclic.

We now discuss   {\it universal} properties of the fluctuations of entropy production.   By applying the integral fluctuation relations~(\ref{eq:strong}) and (\ref{eq:strongCon})   to different examples of stopping times~$T$,  we will derive the following generic relations for the fluctuations of entropy production:

\begin{itemize}
\item In the simple case $T=t$,  where $t$ is a deterministic fixed time, relation (\ref{eq:strong2}) reads $\langle  S_{\rm tot}(t)\rangle \geq 0$, which is a well-known second-law like relation derived in stochastic thermodynamics~\cite{seifert2005entropy, seifert2012stochastic},  and the relation (\ref{eq:strong}) is the stationary integral fluctuation relation $\langle e^{-S_{\rm tot}(t)}\rangle =1$ \cite{seifert2005entropy, seifert2012stochastic}.   The integral fluctuation relation at fixed times implies that in a nonequilibrium process events of negative entropy production must exist and  their likelihood is bounded by \cite{seifert2012stochastic}
\begin{eqnarray}
\mathbb{P}\left(S_{\rm tot}(t)\leq  -s \right) \leq  e^{-s}, \quad s\geq 0 ,\label{eq:markovxa}
\end{eqnarray} 
where $\mathbb{P}\left(\cdot\right)$ denotes the probability of an event.

\item Our  second choice of   stopping times are first-passage times  $T_{\rm fp} = {\rm inf}\left\{t: S_{\rm tot}(t) \notin (-s_-,s_+)\right\}$ for entropy production with  two absorbing boundaries at $-s_-\leq 0$ and $s_+\geq 0$.
As we show in this paper, the integral fluctuation relation Eq.~(\ref{eq:strong}) implies that  the splitting  probabilities $p_- = \mathbb{P}\left[S_{\rm tot}(T) \leq -s_-\right]$ and $p_+ = \mathbb{P}\left[S_{\rm tot}(T) \geq  s_+\right]$ are bounded by
\begin{eqnarray}
p_+ \geq  1 - \frac{1}{e^{s_-}-e^{-s_+}} ,  \quad   p_- \leq    \frac{1}{e^{s_-}-e^{-s_+}}. \label{eq:split2}
\end{eqnarray}   
If the  trajectories of entropy production are continuous, then \cite{neri2017statistics}
\begin{eqnarray}
p_+ =  \frac{e^{s_-}-1}{e^{s_-}-e^{-s_+}}, \quad p_- =  \frac{1-e^{-s_+}}{e^{s_-}-e^{-s_+}}. \label{eq:split1}
\end{eqnarray}

\item    Global infima of entropy production, $S_{\rm inf} = {\rm inf}_{t\geq 0}\:S_{\rm tot}(t)$,  quantify  fluctuations of   negative entropy production.    The cumulative distribution $\mathbb{P}\left[S_{\rm inf} \leq  -s \right]$  is equal to the splitting probability $p_-$ in the limit $s_+\rightarrow \infty$.   Using (\ref{eq:split2}) we obtain \cite{neri2017statistics}
\begin{eqnarray}
\mathbb{P}\left[S_{\rm inf} \leq  -s \right] \leq e^{-s}, \quad s\geq 0, \label{eq:cumInf}
\end{eqnarray} 
which implies the  {\it infimum law} $\langle S_{\rm inf}\rangle \geq -1$ \cite{neri2017statistics}.     It is insightful to compare the two relations (\ref{eq:markovxa}) and (\ref{eq:cumInf}).    
Since $S_{\rm inf} \leq S_{\rm tot}(t)$, the inequality (\ref{eq:cumInf}) implies the inequality (\ref{eq:markovxa}), and (\ref{eq:cumInf}) is thus a stronger result.    Remarkably, the bound~(\ref{eq:cumInf})  is tight for continuous stochastic processes.  Indeed, using (\ref{eq:split1})   we obtain the probability density function  for global infima of the entropy production in continuous processes~\cite{neri2017statistics}, 
\begin{eqnarray} 
p_{S_{\rm inf}}(-s) = e^{-s}, \quad s\geq 0. \label{eq:globInf}
\end{eqnarray}
The mean global infimum is thus $\langle S_{\rm inf}\rangle = -1$.

\item    The survival probability $p_{\rm sur}(t) $ of the  entropy production is the probability that entropy production has not reached a value $s_0$ in the time interval $[0,t]$.       For continuous stochastic processes we obtain the generic expression
\begin{eqnarray}
p_{\rm sur}(t) = \frac{e^{-s_0}  - 1}{e^{-s_0} - \langle e^{-S_{\rm tot}(t)} \rangle_{\rm sur}  }, \label{eq:survxxxxxx}
\end{eqnarray} 
where $\langle \dots \rangle_{\rm sur} $ is an average over trajectories that have not reached the absorbing state in the interval $[0,t]$.

\item   We consider the statistics of the    number of times, $N_{\times}$, that  entropy production crosses  the interval $[-\Delta ,\Delta]$ from $-\Delta$ towards $\Delta$  in one realisation of the process~$\vec{X}(t)$.   The probability of $N_{\times}$ is bounded by 
\begin{eqnarray}
\mathbb{P}(N_\times=0;\Delta)&\ge& 1- e^{-\Delta}, \nonumber\\
\mathbb{P}(N_\times>n)&\le& e^{-\Delta (2n+1)}. \label{ineq:crossing}
\end{eqnarray}
In other words, the probability of observing a large number of crossing decays  at
least exponentially in $N_\times$.
For continuous stochastic processes we obtain a generic expression for the probability  of  $N_\times$ \cite{pigolotti2017generic}, given by
\begin{equation}
\mathbb{P}(N_\times=n)= \left\{
\begin{array}{lr}
1-e^{-\Delta} & n= 0, \\
2\sinh ( \Delta) e^{-2 n\Delta} &n\ge 1.
\end{array}\right. \label{eq:crossing}
\end{equation}
\end{itemize} 
Remarkably, all these results on universal fluctuation properties  are direct  consequences of the integral fluctuation relation for entropy production at stopping times Eq.~(\ref{eq:strong}) and its conditional version, Eq.~(\ref{eq:strongCon}).      

Some of the  results in this paper have already appeared before in the literature or are closely related to existing results. A fluctuation relation analogous to (\ref{eq:strong}) has been derived for the exponential of the negative housekeeping heat, see Eq.~(6) in \cite{chetrite2018martingale}. Since for stationary processes the housekeeping heat is equal to the entropy production, the relation (6) in \cite{chetrite2018martingale} implies the relation (\ref{eq:strong}) in this paper. The relations (\ref{eq:split1}), 
(\ref{eq:cumInf}), (\ref{eq:globInf}) and (\ref{eq:crossing})  have been derived before in  \cite{neri2017statistics} and \cite{pigolotti2017generic}.   Instead, the relations (\ref{eq:strongCon}),  (\ref{eq:QT}), (\ref{eq:carnotStoppxxxx}),  (\ref{eq:split2}), (\ref{eq:survxxxxxx}), and (\ref{ineq:crossing}) are, to the best of our knowledge, derived here for the first time.  Moreover, we demonstrate that all the results (\ref{eq:strongCon}), (\ref{eq:QT}), (\ref{eq:carnotStoppxxxx}), (\ref{eq:split2}), (\ref{eq:split1}), (\ref{eq:cumInf}), (\ref{eq:globInf}), (\ref{eq:survxxxxxx}) and (\ref{eq:crossing}) descend from the integral fluctuation relation  fluctuation relations (\ref{eq:strong}) and (\ref{eq:strongCon})  in a few simple steps, and we discuss the physical meaning of the results derived in this paper on     examples of simple nonequilibrium systems.

The paper is organised as follows.   Section \ref{Sec:notation} introduces the notation used in the paper.   In Section~\ref{sec:casino},  we revisit the theory of  martingales in the context of gambling. In Section~\ref{Sec:StochThermo}, we briefly recall the theory of stochastic thermodynamics, focusing on the aspects we will use in this paper.   These two sections only contain review material, and can be skipped by readers who want to directly read  the new results of this paper.      In  Section~\ref{Sec:Results},  we derive the first important  results of this paper:   the integral fluctuation relations at stopping times~(\ref{eq:strong}) and (\ref{eq:strongCon}).    In Section~\ref{sec:second}, we derive the second law of thermodynamics at stopping times~(\ref{eq:strong2}), and we discuss the physical implications of this  law. In Section \ref{Sec:consequences}, we use the integral fluctuation relation at stopping times to derive universal  properties for the fluctuations of  entropy production in nonequilibrium stationary states, including the relations (\ref{eq:split2})-(\ref{eq:crossing}).      In Section~\ref{sec:finiteStat}, we discuss the effect of finite statistics  on the integral fluctuation relation at stopping times, which is relevant for  experimental validation.   In Section~\ref{Sec:illustration}, we illustrate the second law at stopping times and the integral fluctuation relation at stopping times  in paradigmatic examples of nonequilibrium stationary states.  
 We conclude the paper with a discussion in Section~\ref{sec:disc}.   In the Appendices, we provide details on important proofs and derivations.

\section{Preliminaries and notation}\label{Sec:notation}

In this  paper, we will consider stochastic processes described by $d$ degrees of freedom $\vec{X}(t) = (X_1(t), X_2(t), \ldots, X_d(t))$.    The time index can be discrete, $t\in \mathbb{Z}$, or continuous,  $t\in\mathbb{R}$.  We denote  the full trajectory of $\vec{X}(t)$ by $ \omega  = \left\{\vec{X}(t)\right\}_{t\in(-\infty, \infty)}$   and  the set of all trajectories by $\Omega$.       
  We call a subset $\Phi$  of $\Omega$ a measurable set, or an event, if we can measure the probability $\mathbb{P}\left[\Phi\right]$ to observe a trajectory $\omega$ in $\Phi$.    The $\sigma$-algebra $\mathscr{F}$ is the set of all    subsets of $\Omega$  that are  measurable.

   The triple $(\Omega, \mathscr{F}, \mathbb{P})$ is a probability space.  
We denote random variables on this probability space in  upper case, e.g.,  $X, Y, Z$, whereas deterministic   variables are written in lower case letters, e.g., $x,y,z$.   An exception is the temperature $\mathsf{T}_{\rm env}$, which is a deterministic  variable.     Random variables are functions defined on $\Omega$, i,e, $X:\omega\rightarrow X(\omega)$.  For simplicity we often omit the $\omega$-dependency in our notation of random variables, i.e., we write $X = X(\omega)$, $Y = Y(\omega)$, etc.   For stochastic process, $\vec{X}(t) = \vec{X}(t,\omega)$ is a function on $\Omega$ that returns the value of $\vec{X}$  at time $t$  in the trajectory $\omega$.   
        The expected value of a random variable $X$  is denoted by $\langle X \rangle$  or $E[X]$ and is defined as the integral $\langle X \rangle = E[X]= \int_{\Omega} d\mathbb{P}\: X(\omega)$ in the probability space $(\Omega, \mathscr{F}, \mathbb{P})$.    We write $p_X(x)$ for the probability density function or probability mass function of $X$, if it exists.       We denote vectors by $\vec{x} = (x_1, x_2, \ldots, x_d)$ and we use the notation $t \wedge \tau= {\rm min}  \left\{t,\tau\right\}$.

       We will consider situations where an experimental observer does not have instantaneously access to the complete trajectory $\omega$  but rather tracks a trajectory  $\left\{\vec{X}(s)\right\}_{s\in[0,t]} = \vec{X}^t_0$ in  a finite time interval.  In this case, the  
 set of measurable events gradually expands as time progresses and a larger part of the trajectory $\omega$ becomes visible.      Mathematically this situation is described by an increasing sequence of  $\sigma$-algebras $\left\{\mathscr{F}_t\right\}_{t\geq 0}$  where $\mathscr{F}_t$ contains all the measurable events  $\Phi$  associated with finite-time trajectories   $\vec{X}^s_0$.      The sequence of sub $\sigma$-algebras $\left\{\mathscr{F}_t\right\}_{t\geq 0}$ of $\mathscr{F}$ is called  the  {\em filtration} generated by the stochastic process $X(t)$ and $(\Omega, \mathscr{F}, \left\{\mathscr{F}_t\right\}_{t\geq 0} \mathbb{P})$ is a filtered probability space.   If time is continuous, then we assume that  $\left\{\mathscr{F}_t\right\}_{t\geq 0}$ is right-continuous, i.e., $\mathscr{F}_t = \cap_{s> t}\mathscr{F}_s$; this implies that  the process $\vec{X}(t)$ consists of continuous trajectories intercepted by a  discrete number of jumps.

We denote by  $E[M(t)|\mathscr{F}_s](\omega)$ the  conditional expectation of a random variable $M(t)$ given a sub-$\sigma$-algebra $\mathscr{F}_s$ of $\mathscr{F}$  \cite{liptser2013statistics, williams1991probability}.   Note that conditional expectations  $E[M(t)|\mathscr{F}_s](\omega)$ are random variables on the measurable space $(\Omega, \mathscr{F}_s)$.   Since $E[M(t)|\mathscr{F}_s](\omega)$ is an expectation value we  also  use  the physics notations $E[M(t)|\mathscr{F}_s] = \langle M(t)|\mathscr{F}_s \rangle$ and $E[M(t)|\mathscr{F}_s] =  \langle M(t)|\vec{X}^s_0   \rangle $.     
\section{Martingales}

\label{sec:casino}
In a first subsection, we introduce martingales within the example of games of chance, to illustrate  how fluctuations of a stochastic process can be studied with martingale theory.   The calculations in this subsection are  similar to those  for the stochastic  entropy production presented in       Sections~\ref{Sec:Results} and \ref{Sec:consequences}, with the difference that game of chances are simpler to analyse, since they consist of sequences of independent and identically distributed random variables.    In a second subsection,  we present a  definition of martingales that applies to stochastic processes in continuous and discrete time, and we discuss the optional stopping theorem.

\subsection{Gambling and martingales}
Games of chance have inspired mathematicians as far back as the 17th century and have laid the foundation for probability theory~\cite{hald2003history}.      A question  that has often been studied is the gambler's ruin problem:
Consider a gambler that enters a casino and tries his/her luck at the roulette.   The gambler plays until he/she has either won a predetermined amount of money or until all the initial stake is lost.    We are interested in the probability of success, or equivalently in the ruin probability of the gambler.

The roulette is a game of chance that consists of a ball that rolls at the edge of a  spinning wheel and falls in one of 37 coloured pockets on the wheel: 18 pockets are coloured in red, 18 pockets are coloured in black, and one is coloured in green.  Before each round of the game, the gambler  puts his/her bet on whether the ball will fall in either a red or a black pocket.   If the gambler's call is correct, then he/she wins an amount of  chips equal to the bet size, otherwise he/she looses the betted chips.   The gambler cannot bet for the green pocket.   The presence of the green pocket biases the game in favour of the casino: if the ball falls in the green pocket, then the casino wins all the bets.  We are interested in the gambler's ruin problem: what is the probability that the gambler loses all of his/her initial stakes before reaching a certain amount of profit?
  
A gambling sequence at the roulette can be formalised as a stochastic process $X(t)$ in discrete time, $t=1,2,3,\ldots$   We define $X(t) = 1$ if the ball falls in a red pocket and  $X(t) = -1$ if the ball falls in a black pocket in the $t$-th round of the game.   If the ball falls in a green pocket we set $X(t) = 0$.     We denote the bets of the gambler by the process $Y(t)$: if the gambler calls for red we set  $Y(t) = 1$ and if the gambler calls for black we set 
 $Y(t) = -1$.     The gambler does not bet on green.   Finally, we assume the bet size of the gambler $b$ is constant.     
 
For an ideal roulette, the random variables $X(t)$ are independently drawn from the distribution 
\begin{eqnarray}
p_{X(t)}(x) = \frac{18}{37} \delta_{x,1} + \frac{18}{37} \delta_{x,-1}  +  \frac{1}{37} \delta_{x,0},
\end{eqnarray} 
where $\delta_{x,y}$ is the Kronecker's delta. 
The gambler's bet  $Y(t) = y(X(0), X(1),\ldots, X(t-1))$ with $y$ a function that defines the gambler's betting system.  
 The gambler's fortune  at time $t$ is the process  
 \begin{eqnarray}
F(t) = n +b \sum^t_{s=1}  (Y(s)X(s) + |X(s)| -1), 
\end{eqnarray} 
where $F(0) = n$ is the  initial stake.

The duration of the game is random.   The gambler plays until  a time  $T_{\rm play}$ when the gambler is either ruined, i.e., $F(T_{\rm play}) \leq 0$,
 or  the gambler's fortune has surpassed for the first time a certain amount  $m$, i.e.,  $F(T_{\rm play})\geq m$.  Clearly we require that $m>n = F(0)$, since otherwise $T_{\rm play} = 0$.
 The ruin probability 
\begin{eqnarray}
p_{\rm ruin}(n) = \mathbb{P}\left[F(T_{\rm play} ) \leq 0\right]
\end{eqnarray}
is  the probability that the gambler loses the game. 

The gambler's fortune is a {\it supermartingale} because it is a bounded stochastic process  satisfying
  \begin{eqnarray}
E\left[F(t)|X(0), X(1), \ldots, X(s)\right] \leq F(s), \label{eq:Sub}
\end{eqnarray}  
for all $s\leq t$ and all $t\geq 0$.  
Relation (\ref{eq:Sub}) means that on average the gambler will inevitably lose money, irrespective of the betting system he/she adopts.

A  gambler whose fortune is expected to decrease may still be tempted to play if the  probability of winning is  high enough.  The  probabilities to win or loose the game depend on the fluctuations of  $X(t)$. 
  In the roulette game, the gambler's fortune $F(t)$ can be represented as a biased random walk on the real line $[0,m]$, which starts at the position $n$, and moves each time step either a distance $-b$ to the left or a distance $b$ to the right.     The probability to make a step to the left is $q = 19/37\approx 0.51$ and the probability to make a step to the right is $1-q = 18/37\approx 0.49$.   Hence, the gambler's fortune is slightly biased to move towards the left where $F(t)<0$.     
     The ruin probability $p_{\rm ruin}(n) = \mathbb{P}\left[F(T_{\rm play}) \leq 0\right]$  solves the recursive equation 
   \begin{eqnarray}
p_{\rm ruin}(n) = q \:p_{\rm ruin}(n-b) + (1-q)\, p_{\rm ruin}(n+b) , \quad n\in[0,m], \label{eq:harmx}
   \end{eqnarray}
   with boundary conditions $p_{\rm ruin}(0) = 1$ and $p_{\rm ruin}(m) = 0$.   
  Instead of solving the relations (\ref{eq:harmx}) we   bound the ruin probability $p_{\rm ruin}(n)$ using the theory of martingales~\cite{doyle1984random}.     We define the process
\begin{eqnarray}
M(t) = \left(\frac{q}{1-q}\right)^{F(t)/b}, \label{eq:martingalFort}
\end{eqnarray}
The processes $M(t)$ is a \textit{martingale} relative to  the process $X(t)$ \cite{liptser2013statistics, williams1991probability}.    Indeed, we say that a bounded  process $M$ is a martingale   if 
  \begin{eqnarray}
E\left[M(t)|X(0), X(1), \ldots, X(s)\right] = M(s), \label{eq:MMart}
\end{eqnarray}  
for all $s<t$.

An important property of martingales is that their   expected value evaluated at a stopping time  $T$ of the process $X$ equals  their  expected value at the initial time~\cite{ liptser2013statistics, williams1991probability, doob1953stochastic}, 
\begin{eqnarray}
E\left[M(T)\right] = E[M(0)]. \label{eq:Dooboptxaxxx}
\end{eqnarray}  
Eq.~(\ref{eq:Dooboptxaxxx}) is known as {\it Doob's optional stopping theorem} and will constitute the main tool in this paper to derive   fluctuation properties of stochastic processes. 
In the present example, since 
\begin{eqnarray}
 E\left[M(0)\right] =  \left(\frac{q}{1-q}\right)^{n/b},
\end{eqnarray} 
and since for $q\geq 0.5$
\begin{eqnarray}
\fl p_{\rm ruin}(n) + (1-p_{\rm ruin}(n))\left(\frac{q}{1-q}\right)^{m/b+1}  \geq E\left[M(T_{\rm play})\right] \geq  p_{\rm ruin}(n)  + (1-p_{\rm ruin}(n))\left(\frac{q}{1-q}\right)^{m/b},\nonumber\\
\end{eqnarray} 
Doob's optional stopping theorem (\ref{eq:Dooboptxaxxx}) implies that
\begin{eqnarray}
\frac{ \left(\frac{q}{1-q}\right)^{m/b+1} - \left(\frac{q}{1-q}\right)^{n/b}}{\left(\frac{q}{1-q}\right)^{m/b+1}-1} \geq  p_{\rm ruin}(n) \geq \frac{\left(\frac{q}{1-q}\right)^{m/b} -  \left(\frac{q}{1-q}\right)^{n/b}}{\left(\frac{q}{1-q}\right)^{m/b}-1}. \label{eq:ruin}
\end{eqnarray} 
Hence, Doob's optional stopping theorem    bounds the gambler's ruin probability.  

The formula (\ref{eq:ruin}) provides useful information for the gambler.  
If we start the game with an initial fortune of $10\pounds$, if we play until our fortune reaches $50\pounds$ and if we bet each game $1\pounds$,  then the chance of loosing our initial stake is in between  $94.8$ and $95.2$ percent.      If, on the other hand, we bet each game $10\pounds$, then the ruin probability is in between $82$  and  $85.6$ percent.  Hence,  the probability of winning increases as a function of  the betting size.   Indeed, since the game is biased in favour of the casino, the best strategy is to reduce the number of betting rounds  to the minimal possible and hope that luck plays in our favour.  After all, the outcome of a single game is almost  fair, since  the odds of winning   a single game are $q = 19/37\approx 0.51$.

\subsection{Martingales and the optional stopping theorem}
\label{app:martingales}
We now discuss martingales and the optional stopping theorem for generic stochastic processes $\vec{X}(t)$  in  discrete or continuous time.  
A {\it martingale} process $M(t)$ with respect to another  process $\vec{X}(t)$ is a real-valued  stochastic process that satisfies the following three properties \cite{liptser2013statistics, williams1991probability}: 
\begin{itemize}
\item $M(t)$ is  $\mathscr{F}_t$-adapted, which means that $M(t)$ is a function on trajectories $\vec{X}(0,\ldots, t)$;
\item $M(t)$ is integrable, 
\begin{eqnarray}
E\left[|M(t)|\right] <\infty,\quad  \forall t\geq 0; \label{eq:integrabl}
\end{eqnarray}
 \item the conditional expectation of $M(t)$ given the $\sigma$-algebra $\mathscr{F}_s$ satisfies the property
 \begin{eqnarray}
E[M(t)|\mathscr{F}_s]= M(s),\quad  \forall s\leq t, \quad {\rm and} \quad    \forall t \geq 0.\label{eq:equalxo}
\end{eqnarray}    
The conditional expectation   of a random variable $M(t)$ given a sub-$\sigma$-algebra $\mathscr{F}_s$ of $\mathscr{F}$ is defined as a  $\mathscr{F}_s$-measurable random variable $E[M(t)|\mathscr{F}_s]$  for which $\int_{\omega\in\Phi}  E[M(t)|\mathscr{F}_s]\:{\rm d}\mathbb{P}  = \int_{\omega\in\Phi}M(t) \:{\rm d}\mathbb{P}$ for all $\Phi \in\mathscr{F}_s$  \cite{liptser2013statistics}.  
\end{itemize}
If  instead of the equality  (\ref{eq:equalxo}) we have an inequality $E[M(t)|\mathscr{F}_s]\ \geq  M(s)$, then we call the process a {\it submartingale}.    If $E[M(t)|\mathscr{F}_s]\leq  M(s)$,
then we call the process a {\it supermartingale}.

Fluctuations of a martingale $M(t)$ can be studied with stopping times. 
Stopping times are the random times when the stochastic process $\vec{X}(t)$ satisfies for the first time a given criterion. Stopping times do not allow for clairvoyance (the stopping criterion cannot anticipate the future) and do not allow for cheating (the criterion does not have access to side information).  Aside these constraints, the stopping rule can be an arbitrary  function of the trajectories of the stochastic process $\vec{X}(t)$.        

Formally, a stopping time $T(\omega)\in[0,\infty]$ of a stochastic process $\vec{X}$ is defined as a  random variable for which $\left\{\omega:T(\omega)\leq t\right\}\in\mathscr{F}_t$ for all $t\geq 0$.    Alternatively, we can also define stopping times as functions on trajectories $\omega = \vec{X}^{\infty}_0$ with the property that  the function $T(\omega)$ does not depend on what happens after the stopping time $T$.  

An important result in martingale theory is {\it Doob's optional stopping theorem}.
There exist different versions of the optional stopping theorem, which differ in the  conditions assumed for the martingale process $M(t)$ and the stopping time $T$.    We discuss the version of the theorem presented as Theorem 3.6 in the  book of  Liptser and  Shiryayev~\cite{liptser2013statistics}.

Let $M(t)$  be a martingale relative to the process $\vec{X}$ and let $T$ be a stopping time on the process $\vec{X}$.    If $M$ is uniformly integrable, then   
\begin{eqnarray}
E\left[M(T)\right] = E\left[M(0)\right] . \label{eq:Doob1}
\end{eqnarray}  
If time is continuous we also require that $M(t)$ is rightcontinuous. 
A stochastic process $M(t)$ is called uniformly integrable if 
\begin{eqnarray}
\lim_{m\rightarrow \infty} {\rm sup}_{t \geq 0} \int |M(t)|  I_{|M(t)|\geq m}\:{\rm d}\mathbb{P} = 0 , 
\end{eqnarray}
where $I_{\Phi}(\omega)$ is the indicator function  defined by
\begin{eqnarray}
I_{\Phi}(\omega) = \left\{ \begin{array}{ccc} 1 &{\rm if}& \omega\in \Phi,\\ 0&{\rm if}&\omega\notin \Phi, \end{array}\right.  \label{eq:ind}
\end{eqnarray}
for all $\omega\in\Omega$ and $\Phi\in\mathscr{F}$.    If $M(t)$ is not uniformly integrable, then (\ref{eq:Doob1}) may not hold.  For example if  $M(t) = W(t)$ with $W(t)$ a  Wiener process on $[0,\infty)$ and $T = {\rm inf}\left\{t\geq 0: M(t) = m\right\}$, then $E\left[W(T)\right] = m \neq  E\left[W(0)\right] = 0$, where we have used the convention that $0\cdot \infty = 0$ \cite{tao2011introduction}. 

An extended version of the optional stopping theorem holds for two stopping times
$T_1$ and $T_2$  with the property $\mathbb{P}[T_2 \leq T_1] = 1$, 
\begin{eqnarray}
E\left[M(T_1)|\mathscr{F}_{T_2}\right](\omega) = M(T_2(\omega),\omega),  \label{eq:Doob2}
\end{eqnarray}
where the  $\sigma$-algebra $\mathscr{F}_{T_2}$ consists of all sets $\Phi\in \mathscr{F}$ such that $\Phi \cap \left\{\omega:T_{2}(\omega)\leq t\right\} \in \mathscr{F}_t $.

\section{Stochastic thermodynamics for stationary processes}\label{Sec:StochThermo}
In this section, we briefly introduce the formalism of stochastic thermodynamics in nonequilibrium stationary states; for reviews see  Refs.~\cite{maes2003origin, sekimoto2010stochastic, jarzynski2011equalities, seifert2012stochastic}.  
We use a probability-theoretic approach~\cite{maes2000definition, jiang2003mathematical, neri2017statistics}, which  has the advantage of dealing with Markov chains, Markov jump processes, and Langevin process in one unified framework. It is moreover the natural language to deal with martingales.

The stochastic entropy production $S_{\rm tot}(t)$ is defined in terms of a  probability measure $\mathbb{P}$ of a stationary stochastic process and its  time-reversed measure  $\mathbb{P}\circ\Theta$.     The time-reversal map $\Theta$, with respect to the origin $t=0$, is a measurable involution on trajectories $\omega$  with the property that 
    $X_i(t, \Theta(\omega)) = X_i(-t,\omega)$ for  variables of even parity with respect to time reversal and $X_i(t, \Theta(\omega)) =  -X_i(-t,\omega)$ for variables of odd parity with respect to time reversal.     We say that the measure $\mathbb{P}$ is stationary if $\mathbb{P} = \mathbb{P}\circ\mathcal{T}_t$ for all $t\in\mathbb{R}$, with $\mathcal{T}_t$ the time-translation map, i.e., $\vec{X}(t',\mathcal{T}_t(\omega)) =\vec{X}(t'+t, \omega)$ for all $t'\in\mathbb{R}$.
In order to define an entropy production  we require that the process $\vec{X}(t,\omega)$  is reversible.  This means that for all finite $t\geq 0$ and for all $\Phi\in\mathscr{F}_t$ it holds that $\mathbb{P}\left[\Phi\right]=0$ if and only if $(\mathbb{P}\circ\Theta)\left[\Phi\right]=0$.   In other words, if an event happens with zero probability in the forward dynamics, then this event also occurs with zero probability  in the time-reversed dynamics.    In probability theory, one says that  $\mathbb{P}$ and $\mathbb{P}\circ\Theta$ are locally mutually absolute continuous.
 
Given the above assumptions, we define the entropy production in a stationary process $\vec{X}$ by~\cite{lebowitz1999gallavotti, maes2000definition, jiang2003mathematical, neri2017statistics}
\begin{eqnarray}
S_{\rm tot}(t,\omega) := \log \frac{\left.{\rm d}\mathbb{P}\right|_{\mathscr{F}_t}}{\left.{\rm d}(\mathbb{P}\circ\Theta)\right|_{\mathscr{F}_t}}(\omega),\quad t\geq0, \quad \omega\in\Omega,  \label{eq:SGeneral}
\end{eqnarray}       
where we have used the Radon-Nikodym derivative of   the restricted measures $\left.\mathbb{P}\right|_{\mathscr{F}_t}$ and $\left.(\mathbb{P}\circ\Theta)\right|_{\mathscr{F}_t}$ on $\mathscr{F}$.     The restriction of a measure $\mathbb{P}$ on a sub-$\sigma$-algebra  $\mathscr{F}_t$ of $\mathscr{F}$ is defined  
 by $\left.\mathbb{P}\right|_{\mathscr{F}_t}[\Phi] = \mathbb{P}[\Phi]$ for all $\Phi\in\mathscr{F}_t $.    
 If $t$ is  continuous, then $S_{\rm tot}(t, \omega)$ is rightcontinuous, since we have assumed the rightcontinuity of the filtration $\left\{\mathscr{F}_t\right\}_{t\geq0}$.
  Local mutual  absolute continuity of the two measures $\mathbb{P}$ and $\mathbb{P}\circ\Theta$ implies that the Radon-Nikodym derivative in (\ref{eq:SGeneral}) exists and is almost everywhere uniquely defined. 
   The definition (\ref{eq:SGeneral})  states that entropy production is  the probability density of the measure $\mathbb{P}$ with respect to the time-reversed measure $\mathbb{P}\,\circ\,\Theta$; it is a functional of  trajectories  $\omega$ of the stochastic process $X$ and characterises their time-irreversibility. 
 
The definition (\ref{eq:SGeneral}) of the stochastic entropy production is general.  It applies to Markov chains, Markov jump processes, diffusion processes, etc.  
For Markov chains, the relation (\ref{eq:SGeneral})  is equivalent to the expression (\ref{eq:Markov0}) for entropy production in terms of  probability density functions of trajectories.  Consider for example the case of $\vec{X}(t)\in\mathbb{R}^d$ and $t\in \mathbb{Z}$ and let us assume for simplicity that all degrees of freedom are of even parity with respect to time reversal.  Using $\left.{\rm d}\mathbb{P}\right|_{\mathscr{F}_t} = p(\vec{X}(1), \vec{X}(2), \ldots, \vec{X}(t))\left.{\rm d}\lambda\right|_{\mathscr{F}_t}$, with $\left.\lambda\right|_{\mathscr{F}_t}$ the Lebesgue measure on $\mathbb{R}^{td}$, the entropy production is indeed of the form given by relation (\ref{eq:Markov0}).
However, formula (\ref{eq:SGeneral}) is  more general than (\ref{eq:Markov0}) because it also applies to cases where the path probability density with respect to a Lebesgue measure does not exist, as is the case with stochastic processes in continuous time.  

 For systems that are weakly coupled to one or more environments in equilibrium, the entropy production (\ref{eq:SGeneral}) is  equal to  \cite{seifert2012stochastic}
\begin{eqnarray}
S_{\rm tot}(t) =  \Delta S_{\rm sys}(t) + S_{\rm env}(t) , \label{eq:Stotxx}
\end{eqnarray}
where $S_{\rm env}(t)$ is the entropy change of the environment, and  where 
\begin{eqnarray}
 \Delta S_{\rm sys}(t) = - \log \frac{p_{\rm ss}(\vec{X}(t))}{p_{\rm ss}(\vec{X}(0))}\label{eq:Ssys}
\end{eqnarray}
is the  system entropy change associated with     the  stationary probability density function 
 $ p_{\rm ss}(\vec{X}(t))$ of $\vec{X}(t)$.

\section{Integral fluctuation relations at stopping times}\label{Sec:Results}
In this section we initiate  the   study of the fluctuations of the entropy production in  stationary processes. We follow 
 an approach  similar 
to the one presented in Section~\ref{sec:casino} for the fluctuations of a gambler's fortune, namely, we first identify a martingale process  related to the   entropy production, which is the exponentiated negative entropy $e^{-S_{\rm tot}(t)}$,  and we then apply Doob's optional stopping theorem (\ref{eq:Doob1}) to this martingale process.   
   Since $e^{-S_{\rm tot}(t)}$ is unbounded, we require uniform integrability of $e^{-S_{\rm tot}(t)}$ in order to apply Doob's optional stopping theorem~(\ref{eq:Doob1}).   Therefore,  we obtain two versions of the integral fluctuation relation at stopping times: a first version holds within finite-time windows and a second version holds when $S_{\rm tot}(t)$ is bounded for all times  $t\leq T$.   These two versions represent two different ways to ensure that the total available entropy in the system and environment is finite.   
   Note that  when we applied Doob's optional stopping theorem in the gambling problem, we also required that the gambler's fortune is finite.     Finally, we obtain  conditional  integral fluctuation relations by applying  the conditional version  (\ref{eq:Doob2}) of Doob's optional stopping theorem to  $e^{-S_{\rm tot}(t)}$.    
\subsection{The  martingale structure of exponential entropy production}
The exponentiated negative entropy production  $e^{-S_{\rm tot}(t)}$ associated with a stationary stochastic process $\vec{X}(t)$ is a martingale process relative to  $\vec{X}(t)$.  
 Indeed,   $S_{\rm tot}(t)$ is a $\mathscr{F}(t)$-adapted process, $E\left[e^{-S_{\rm tot}(t)}\right] =1$, and    in \ref{Append1}  we show that   \cite{chetrite2011two, neri2017statistics}
\begin{eqnarray}
E\left[e^{-S_{\rm tot}(t)}|\mathscr{F}_s\right] = e^{-S_{\rm tot}(s)}, \quad  s\leq t. 
\end{eqnarray}
As a consequence, entropy production is a submartingale: 
\begin{eqnarray}
E\left[S_{\rm tot}(t)|\mathscr{F}_s\right]  \geq S_{\rm tot}(s), \quad  s\leq t. 
\end{eqnarray}
Notice that we can draw an analogy between thermodynamics and gambling by identifying 
 the negative entropy production with a gambler's fortune and by identifying 
the exponential $e^{-S_{\rm tot}(t)}$ with the martingale (\ref{eq:martingalFort}). 
 
 \subsection{Fluctuation relation at stopping times within a finite-time window}
 We  apply Doob's optional stopping theorem~(\ref{eq:Doob1}) to the martingale $e^{-S_{\rm tot}(t)}$.    We consider first
the case  when an experimental observer measures a  stationary stochastic processes $\vec{X}(t)$ within a finite-time window $t\in[0,\tau]$.     In this case, the experimental observer  measures in fact  the process $e^{-S_{\rm tot}(t\wedge \tau)}$, where we have used the notation 
\begin{equation}
t \wedge \tau = {\rm min}  \left\{t,\tau\right\}.
\end{equation}
  The process $e^{-S_{\rm tot}(t\wedge \tau)}$ is uniformly integrable, as we show in the \ref{Append1}, and therefore
   \begin{eqnarray}
   \left\langle e^{-S_{\rm tot}(T\wedge \tau)}\right\rangle = 1, \label{stoppingInt}
   \end{eqnarray}    
   holds for all stopping times $T$  of $\vec{X}(t)$.  
  
  \subsection{Fluctuation relation at stopping times within an infinite-time window} 
  We discuss an integral fluctuation relation for stopping times within an infinite-time window, i.e., $T\in[0,\infty]$.   In the  \ref{Append2} we prove that if the conditions
   \begin{enumerate}[(i)]  
  \item  $e^{-S_{\rm tot}}$ converges $\mathbb{P}$-almost surely to $0$ in the limit $t\rightarrow \infty$   
  \item  $S_{\rm tot}(t)$ is  bounded  for all $t\leq T$
      \end{enumerate}
      are met, then   \begin{eqnarray}
   \left\langle e^{-S_{\rm tot}(T)}\right\rangle = 1.\label{stoppingIntxxx}
   \end{eqnarray}     
      The condition (i) is a  reasonable assumption  for nonequilibrium stationary states, since $\langle S_{\rm tot}(t)\rangle$ grows extensive in time as $\langle S_{\rm tot}(t)\rangle = \sigma t$ with $\sigma$ a  positive number.    The condition (ii) can be imposed on $T$ by considering   a 
      stopping time $T\wedge T_{\rm fp}$  where $T_{\rm fp} = {\rm inf}\left\{t: S_{\rm tot}(t) \notin (-s_-,s_+)\right\}$  is a first-passage time with two thresholds $s_-, s_+ \gg 1$, which can be considered  large compared to the typical values of entropy production at the stopping time~$T$.

   \subsection{Conditional integral fluctuation relation at stopping times} 
We can also apply the conditional optional stopping-time theorem (\ref{eq:Doob2}) to $e^{-S_{\rm tot}(t)}$.   We then obtain the conditional integral fluctuation relation (\ref{eq:strongCon}) for stopping times $T_2\leq T_1$, viz., 
\begin{eqnarray}
\left\langle e^{-S_{\rm tot}(T_1)}|\mathscr{F}_{T_2}\right\rangle =    e^{-S_{\rm tot}(T_2)}.   \label{stoppingInt2}
\end{eqnarray}
 The relation (\ref{stoppingInt2}) is valid either for finite stopping times $T_1\in [0,\tau]$  or for stopping times $T_1$ and $T_2$ for which $S_{\rm tot}(t)$ is bounded for all $t\in [0,T_1]$.
The integral fluctuation relations at stopping times (\ref{stoppingInt}),  (\ref{stoppingIntxxx}) and  (\ref{stoppingInt2})   imply that certain stochastic properties of entropy production are universal.     This will be discussed in the next section.

\section{Second law of thermodynamics at stopping times}\label{sec:second}
Jensen's inequality $   \langle e^{-S_{\rm tot}(T\wedge \tau)}\rangle \geq  e^{-\langle S_{\rm tot}(T\wedge \tau) \rangle} $ together with the integral fluctuation relation  (\ref{stoppingInt})  imply that
\begin{eqnarray}
\langle S_{\rm tot}(T\wedge \tau)\rangle \geq 0. \label{strong}
\end{eqnarray}  
The relation (\ref{strong}) states that on average entropy production always increases, even when we stop the process at a random time $T$ chosen according to a given protocol.    This law is akin to the relation  (\ref{eq:Sub}) describing that a gambler cannot make profit out of a fair game of chance, even when he/she quits the game in an intelligent manner.    Analogously,  (\ref{stoppingIntxxx})  implies the  law  $\langle S_{\rm tot}(T)\rangle \geq 0$ for unbounded stopping times $T$.    We have thus derived the second law (\ref{eq:strong2})  of thermodynamics at stopping times.  

When applying this second law to examples of physical processes  one obtains  interesting bounds on heat and work in nonequilibrium stationary states.  Below we first discuss  bounds on the average dissipated heat in isothermal processes and then bounds on the average work in stationary stochastic heat engines.

\subsection{Bounds on heat absorption in isothermal processes}
For  systems that are in contact  with one thermal reservoir at temperature $\mathsf{T}_{\rm env}$  and for which the entropy of hidden internal degrees of freedom is negligible, the entropy production is given by (\ref{eq:Stotxx}-\ref{eq:Ssys}), and the    environment entropy takes the form \cite{seifert2012stochastic, seifert2008stochastic}
\begin{eqnarray}
S_{\rm env}(t)  = -\frac{Q(t)}{\mathsf{T}_{\rm env}}, \label{eq:thermod}
\end{eqnarray}   
where $Q$ is the heat transferred from  the environment to the system.
Relation  (\ref{eq:thermod}) relates the stochastic entropy production~(\ref{eq:Stotxx})  to the stochastic heat that enters into the first law of thermodynamics.     Therefore, for isothermal systems the second law (\ref{strong}) reads 
\begin{eqnarray}
\langle Q (T\wedge \tau)\rangle  \leq \mathsf{T}_{\rm env}  \langle  \Delta S_{\rm sys}(T\wedge \tau)\rangle = \mathsf{T}_{\rm env}  \Big\langle \log  \frac{p_{\rm ss}(X(0))}{ p_{\rm ss}(X(T\wedge \tau))} \Big\rangle. \label{eq:second2}
\end{eqnarray}  
Analogously,  we  obtain the relation (\ref{eq:QT}) for unbounded stopping times $T$.     

The relation (\ref{eq:second2}) implies that  it is not possible to extract on average heat from a thermal reservoir 
when the system state is invariant.   Indeed, when $X(T\wedge \tau) = X(0)$, then $\langle Q (T\wedge \tau)\rangle\leq 0$.       However, if the system entropy at the stopping time $T$ is different than the entropy in the stationary state, then it is possible to extract on average at most an amount  $\mathsf{T}_{\rm env}  \Big\langle \log  \frac{p_{\rm ss}(X(0))}{ p_{\rm ss}(X(T\wedge \tau))} \Big\rangle$ of heat from the thermal reservoir.    

For  systems in equilibrium $p_{\rm ss}(x) \sim e^{-v(x)/\mathsf{T}_{\rm env}}$, such that  the bound (\ref{eq:second2})  reads $\langle Q (T\wedge \tau)\rangle \leq  \langle \Delta v(T\wedge \tau)\rangle $.  Moreover, according to the first law of thermodynamics $\langle Q (T\wedge \tau)\rangle =  \langle \Delta v(T\wedge \tau)\rangle $, such that  for systems in equilibrium the  bound (\ref{eq:second2}) is tight.

Notice that the bound on the right hand side of (\ref{eq:second2}) is maximal for stopping times of the form 
\begin{eqnarray}
T^\dagger = {\inf}\left\{ t\geq 0: p_{\rm ss}(X(t)) =  {\rm min}_{x\in\mathcal{X}} p_{\rm ss}(x) \right\} . \label{eq:minT}
\end{eqnarray} 
If $X(t)$ is a recurrent process, i.e.~$\mathbb{P}(T^\dagger<\infty)=1$,  and if $\tau\rightarrow \infty$, then  
\begin{eqnarray}
 \langle Q (T^\dagger)\rangle  \leq    \mathsf{T}_{\rm env}  \sum_{x\in \mathcal{X}}p_{\rm ss}(x)  \log  \frac{p_{\rm ss}(x)}{{\rm min}_{x'\in\mathcal{X}} p_{\rm ss}(x')}  . \label{eq:boundHeat}
\end{eqnarray}

\subsection{Efficiency of heat engines at stopping times:  the case of Feynman's ratchet}
 We consider stochastic  heat engines in contact with two thermal baths at temperatures $\mathsf{T}_{\rm c}$ and $\mathsf{T}_{\rm h}$ with $\mathsf{T}_h \geq \mathsf{T}_c$.   A paradigmatic example is  Feynman's ratchet   \cite{feynman1965feynman, parrondo1996criticism, sekimoto1997kinetic,magnasco1998feynman,singh2017feynman}, which is composed of a ratchet wheel with a pawl that is   mechanically linked   by an axle  to  a vane.   The ratchet wheel and the pawl are immersed in a hot thermal reservoir, and the vane is immersed in a cold thermal reservoir.    An external mass  is connected to the axle of the Feynman ratchet and follows the movement of the ratchet wheel.     If the wheel turns in the clockwise direction then the axle performs work on the mass, whereas if the wheel turns in the counterclockwise direction then the mass performs work on the axle . 
 
 We now perform an analysis of the Feynman ratchet at stopping times.    For example, we  ask the question what are the efficiency and the power of the ratchet when the system is stopped right before or after the "main event", i.e., the passage of the pawl over the peak of the ratchet wheel.    The first law of thermodynamics implies that 
   \begin{eqnarray}
  \langle Q_{\rm h}(T)\rangle + \langle Q_{\rm c}(T)\rangle + \langle W(T)\rangle = \langle \Delta v(T)\rangle, 
  \end{eqnarray}
  where $Q_{\rm h}$ is the heat absorbed by the ratchet from the hot reservoir,  $Q_{\rm c}$ is heat   absorbed by the ratchet from the cold reservoir,   $v$ is the mechanical energy stored in the pawl, and $W$ is the work performed on the external mass.  For the Feynman's ratchet,  the second law of thermodynamics at stopping times $T$ reads 
   \begin{eqnarray}
 \frac{\langle Q_{\rm h}(T)\rangle}{\mathsf{T}_{\rm h}} +   \frac{\langle Q_{\rm c}(T)\rangle}{\mathsf{T}_{\rm c}}  \leq \langle \Delta S_{\rm sys}(T)\rangle,
  \end{eqnarray}
  where $S_{\rm sys}$ is the entropy of the ratchet.    If $\langle Q_{\rm h}(T)\rangle>0$, the first and second law of thermodynamics at stopping times imply the inequality (\ref{eq:carnotStoppxxxx}), i.e.,
     \begin{eqnarray} 
\eta_{T}\    \leq    \eta_{\rm C} -  \frac{\langle \Delta \mathcal{F}_{\rm c}(T)\rangle}{\langle Q_{\rm h}(T)\rangle }, \nonumber
  \end{eqnarray} 
  where we have introduced the efficiency at stopping times
  \begin{eqnarray}
\eta_{T} := - \frac{\langle W(T)\rangle}{\langle Q_{\rm h}(T)\rangle}, \nonumber
  \end{eqnarray}
the Carnot efficiency 
  \begin{eqnarray}
  \eta_{\rm C}= 1-\frac{ \mathsf{T}_{\rm c}}{\mathsf{T}_{\rm h}},
  \label{eq:Carnotefficiency}
  \end{eqnarray}
  and the system free energy
    \begin{eqnarray}
   \langle \Delta \mathcal{F}_{\rm c}(t)\rangle = \langle \Delta v(t)\rangle -\mathsf{T}_{\rm c}\langle \Delta S_{\rm sys}(t)\rangle.
    \end{eqnarray} 
For   $T = t$,  $\langle \Delta F_{\rm c}(T) \rangle =  0 $ and we obtain the classical  Carnot bound $\eta_t\leq \eta_{\rm C}$.       Moreover, if $X(T) = X(0)$, which implies that the process stops when it returns to its original state, then $\langle \Delta F_{\rm c}(T) \rangle =  0 $ and we obtain again the classical  Carnot bound $\eta_t\leq \eta_{\rm C}$.   Hence, it is not possible to exceed on average the Carnot efficiency when the final state equals the initial state and thus when the heat engine is  a  cyclic process in phase space.  

However, for general stopping times $T$,  $\langle \Delta \mathcal{F}_{\rm c}(T)\rangle$  is different than zero.             Interestingly,   for stopping times $T$ for which  $   \langle \Delta \mathcal{F}_c(T)\rangle/\langle Q_{\rm h}(T) \rangle$  is negative,  the second law of thermodynamics at stopping times implies that  $\eta_T$  is bounded by a constant that is larger than the Carnot efficiency.    
      Note that the   stopping-time efficiency $\eta_T$ is defined as the ratio of averages, and not as the average of the ratios. In general $\langle W(T)\rangle/\langle Q_{\rm h}(T)\rangle \neq \langle W(T) / Q_{\rm h}(T) \rangle$, and the latter  corresponds to the average of an unbounded random variable whose value at fixed times $T=t$ has been previously studied in~\cite{sekimoto2000carnot,verley2014unlikely,verley2014universal,polettini2015efficiency,martinez2016brownian}.  

      
    
    Another interesting property of thermodynamic observables at stopping times is that they can  take a different sign with respect to their stationary averages.  For example, it is possible that at the stopping time $T$  the fluxes of the Feynman ratchet have the same sign as those in a refrigerator, namely, 
 $\langle W(T)\rangle > 0$,  $\langle Q_{\rm c}(T)\rangle >0$ and $\langle Q_{\rm h}(T)\rangle < 0$.    To evaluate the performance of this process, we introduce the coefficient   $\epsilon_T:=- \langle Q_{\rm c}(T)\rangle / \langle W(T)\rangle>0$, for which  the second law of thermodynamics at stopping times reads $\epsilon_T  \leq \left[\left(\frac{\mathsf{T}_{\rm h}}{\mathsf{T}_{\rm c}}-1\right)+\frac{\langle \Delta \mathcal{F}_{\rm h}(T)\rangle}{\langle Q_{\rm c}(T)\rangle}\right]^{-1}$,
with $ \langle \Delta \mathcal{F}_{\rm h}(T)\rangle := \langle \Delta v(T)\rangle -\mathsf{T}_{\rm h}\langle \Delta S_{\rm sys}(T)\rangle$.  For    fixed times $T=t$  and for  stopping times $T$ with $X(T)=X(0)$,  we recover the classical bound $\epsilon_T \leq \mathsf{T}_{\rm c}/(\mathsf{T}_{\rm h}-\mathsf{T}_{\rm c})$~\cite{velasco1997new,de2012optimal,rana2016anomalous}.

       In Section~\ref{Sec:illustration}, 
  we  illustrate the bounds (\ref{eq:QT}) and (\ref{eq:carnotStoppxxxx})   on  simple physical models.

\section{Universal properties of stochastic entropy production} \label{Sec:consequences}
We use the fluctuation relations (\ref{stoppingInt}), (\ref{stoppingIntxxx}) and   (\ref{stoppingInt2}) to derive universal relations about the stochastic properties of entropy production in stationary  processes.     

All our results hold  for  nonequilibrium stationary states for which  $\lim_{t\rightarrow \infty}e^{-S_{\rm tot}(t)} = 0$ ($\mathbb{P}$-almost surely).

If  entropy production is bounded for  $t<T$, then  we will use the optional stopping  theorem (\ref{stoppingIntxxx}) for stopping times within  infinite-time windows , and  if entropy production  is unbounded for $t<T$, then we will use the optional stopping theorem (\ref{stoppingInt}) for stopping times within finite-time windows.

   \subsection{Fluctuation properties at fixed time $T=t$} 

We first consider the case where the stopping time $T$ is a fixed non-fluctuating time $t$, i.e.,  $T=t$.  In this case the fluctuation relation (\ref{stoppingInt}) is  the integral fluctuation theorem derived in  \cite{seifert2005entropy}, 
   \begin{eqnarray}
   \langle e^{-S_{\rm tot}(t)}\rangle = 1, \label{stoppingIntx}
   \end{eqnarray}
   and the second law inequality (\ref{strong}) yields the second law of stochastic  thermodynamics~\cite{seifert2012stochastic} 
   \begin{eqnarray}
\langle S_{\rm tot}(t)\rangle \geq 0. \label{weak}
\end{eqnarray}
The relation (\ref{stoppingIntx}) provides a bound on negative  fluctuations of entropy production, and for isothermal systems bounds on  the fluctuations of work~\cite{jarzynski2011equalities}.   Since $e^{-S_{\rm tot}(t)}$ is a positive random variable, we can use Markov's inequality, see e.g.~chapter 1 of  \cite{tao2012topics}, to bound events of large $e^{-S_{\rm tot}(t)}$, namely, 
\begin{eqnarray}
\mathbb{P}\left(e^{-S_{\rm tot}(t)}\geq \lambda \right) \leq \frac{ \langle e^{-S_{\rm tot}(t)}\rangle}{\lambda },\quad \lambda \geq 0. \label{eq:markovxxx}
\end{eqnarray} 
Using the integral fluctuation relation (\ref{stoppingIntx}) together with  (\ref{eq:markovxxx}) we obtain
\begin{eqnarray}
\mathbb{P}\left(S_{\rm tot}(t)\leq  -s \right) \leq  e^{-s}, \quad s\geq 0 ,\label{eq:markov}
\end{eqnarray} 
which is a well-known bound on the probability of negative entropy production, see Equation (54) in \cite{seifert2012stochastic}.

The relations (\ref{stoppingInt}) and (\ref{strong}) are more general than the relations (\ref{stoppingIntx}) and (\ref{weak}), since the former concern an average over an ensemble of trajectories of variable length $x^T_0$ whereas the latter concern  an average over an ensemble of trajectories of fixed length $x^t_0$.  Therefore, we  expect that evaluating (\ref{stoppingInt}) at fluctuating stopping times $T$ it is possible to derive stronger constraints  on the probability of negative entropy production than (\ref{eq:markov}).   This is the program we will pursue in the following sections.

\subsection{Splitting probabilities of entropy production} \label{sec:int}
We consider the first-passage time  
\begin{eqnarray}
T_{\rm fp} = {\rm inf}\left\{t\geq 0: S_{\rm tot}(t) \notin (-s_-,s_+)\right\}
\end{eqnarray}
for entropy production with two absorbing thresholds at $-s_-<0$ and $s_+>0$.   If the set $\left\{t: S_{\rm tot}(t) \notin (-s_-,s_+)\right\}$ is empty, then $T_{\rm fp} =  \infty$.   Since entropy production is bounded for all values $t<T_{\rm fp}$ we can use the optional stopping theorem (\ref{stoppingIntxxx}).  
  
We split the ensemble of trajectories $\Omega$ into two sets $\Omega_- = \left\{S_{\rm tot}(T_{\rm fp})
\leq -s_-\right\}$ and $\Omega_+ = \left\{S_{\rm tot}(T_{\rm fp})
\geq s_+\right\}$.  Since    $S_{\rm tot}(T_{\rm fp})\in (\infty,-s_-] \cup [s_+,\infty)$ and $\lim_{t\rightarrow \infty}e^{-S_{\rm tot}(t)} = 0$ ($\mathbb{P}$-almost surely),   the splitting probabilities  $p_- = \mathbb{P}\left[\Omega_- \right]$ and $p_+ = \mathbb{P}\left[\Omega_+ \right]$ have a total probability  of one,  
 \begin{eqnarray}
 p_+ + p_- = 1 .\label{eq:norm}
 \end{eqnarray} 
We  apply the integral fluctuation relation  (\ref{stoppingIntxxx}) to the stopping time $T_{\rm fp}$:
\begin{eqnarray}
  1 = p_- \langle e^{-S_{\rm tot}(T_{\rm fp})}\rangle_{-}+ p_+\langle e^{-S_{\rm tot}(T_{\rm fp})}\rangle_{+} \label{eq:fluct}
\end{eqnarray} 
where  
\begin{eqnarray}
\langle e^{-S_{\rm tot}(T_{\rm fp})} \rangle_+ &=&  \frac{\int_{\omega\in\Omega_+} d\mathbb{P} \:e^{-S_{\rm tot}(T_{\rm fp}(\omega), \omega)}}{\int_{\omega\in\Omega_+} d\mathbb{P} }, \\ 
\langle e^{-S_{\rm tot}(T_{\rm fp})} \rangle_- &=&  \frac{\int_{\omega\in\Omega_-} d\mathbb{P} \:e^{-S_{\rm tot}(T_{\rm fp}(\omega), \omega)}}{\int_{\omega\in\Omega_-} d\mathbb{P}}.
\end{eqnarray} 
The relations (\ref{eq:norm}) and (\ref{eq:fluct}) imply  that
\begin{eqnarray}
p_+ &=& \frac{\langle e^{-S_{\rm tot}(T)}\rangle_- -1}{\langle e^{-S_{\rm tot}(T)} \rangle_- - \langle e^{-S_{\rm tot}(T)} \rangle_+}, \label{relsplit1}\\ 
p_- &=& \frac{1 - \langle e^{-S_{\rm tot}(T)}\rangle_+ }{\langle e^{-S_{\rm tot}(T)} \rangle_- - \langle e^{-S_{\rm tot}(T)} \rangle_+}.  \label{relsplit2}
\end{eqnarray}
Moreover, since $S_{\rm tot}(T_{\rm fp})\in (\infty,-s_-] \cup [s_+,\infty)$ we have that  $\langle e^{-S_{\rm tot}(T)}\rangle_- \geq e^{s_-}$ and $\langle e^{-S_{\rm tot}(T)}\rangle_+ \leq e^{-s_+}$ .   Using these two inequalities in (\ref{relsplit1}) and (\ref{relsplit2}) we obtain the  universal inequalities  (\ref{eq:split2}) for the splitting probabilities, viz.,  
 \begin{eqnarray}
p_+ \geq  1 - \frac{1}{e^{s_-}-e^{-s_+}} , \quad   p_- \leq    \frac{1}{e^{s_-}-e^{-s_+}}, \nonumber
\end{eqnarray}
which hold for first-passage times with $s_-,s_+>0$.
 
In the case where $\vec{X}(t)$ is a continuous stochastic process $S_{\rm tot}(T_{\rm fp})\in  \left\{-s_-,s_+\right\}$ holds with probability one.
 Using in (\ref{relsplit1}) and (\ref{relsplit2}) that  $\langle e^{-S_{\rm tot}(T)}\rangle_- = e^{s_-}$ and $\langle e^{-S_{\rm tot}(T)}\rangle_+ = e^{-s_+}$,  we obtain
   \begin{eqnarray}
p_+ =   \frac{e^{s_-}-1}{e^{s_-}-e^{-s_+}}, \quad   p_- =    \frac{1-e^{-s_+}}{e^{s_-}-e^{-s_+}},  \nonumber
\end{eqnarray}
which are the relations (\ref{eq:split1}). 
Hence, the splitting probabilities of entropy production are universal for continuous processes.  

The bounds  (\ref{eq:split2}) apply not only to first-passage times but  hold more generally for stopping times $T$ of $\vec{X}(t)$ for which  $S_{\rm tot}(T)\in (\infty,-s_-] \cup [s_+,\infty)$ holds with probability one and for which $S_{\rm tot}(t)$ is bounded for all $t\in [0, T]$.     Analogously, the equalities (\ref{eq:split1}) apply not only to first-passage times but hold more generally for stopping times $T$ of $\vec{X}(t)$ for which  $S_{\rm tot}(T)\in \left\{-s_-,s_+\right\}$ holds with probability one and for which  $S_{\rm tot}(t)$ is bounded for all $t\in [0, T]$. 

\subsection{Infima of entropy production  } \label{inf:Sub} 
  We can use the  results of the previous subsection to derive universal bounds and equalities for the statistics of infima of entropy production.      The  global infimum of entropy production is defined by 
  \begin{eqnarray}
  S_{\rm inf} = {\rm inf}_{t\geq 0}S_{\rm tot}(t)
  \end{eqnarray}
  and denotes the most negative value of entropy production.   

Consider  the first-passage time $T_{\rm fp}$, denoting the first time when entropy production either goes below $-s_-$ or goes above $s_+$, with $s_-, s_+\geq0$ and its associated splitting probability $p_-$  The cumulative distribution of $S_{\rm inf}$ is given by 
\begin{eqnarray}
\mathbb{P}\left[S_{\rm inf} \leq -s_-\right] = \lim_{s_+\rightarrow \infty}p_- .\label{eq:relationCump-}
\end{eqnarray} 
Using the inequalities (\ref{eq:split2}) we obtain the bound (\ref{eq:cumInf}) for the cumulative distribution of infima of entropy production, i..e., 
\begin{eqnarray}
\mathbb{P}\left[S_{\rm inf} \leq -s_-\right] \leq e^{-s_-}. \nonumber
\end{eqnarray}
 The inequality (\ref{eq:cumInf})   bears a strong similarity with the inequality (\ref{eq:markov}).  However, since  by definition 
 $S_{\rm inf} \leq S_{\rm tot}(t)$ for any value of $t$, we also have that $\mathbb{P}\left[S_{\rm tot} \leq -s_-\right]    \leq \mathbb{P}\left[S_{\rm inf} \leq -s_-\right]$ and therefore the inequality   (\ref{eq:cumInf}) is stronger.

For continuous processes $\vec{X}$, the inequality (\ref{eq:cumInf}) becomes an equality.  Indeed, if the stochastic process $\vec{X}$ is continuous, then with probability one   $S_{\rm tot}(T_{\rm fp})\in \left\{-s_-, s_+\right\}$, and therefore $p_-$ is given by (\ref{eq:split1}).   Using  the relation (\ref{eq:relationCump-}),  we obtain  
\begin{eqnarray}
\mathbb{P}\left[S_{\rm inf} \leq -s_-\right] = e^{-s_-}. \label{eq:Sinfxx}
\end{eqnarray} 
As a consequence, for continuous stochastic processes the global infimum  $S_{\rm inf}$ of entropy production is characterised by an exponential  probability density (\ref{eq:globInf}) with mean value mean value $\langle S_{\rm inf} \rangle = -1$.

\subsection{Survival probability of entropy production}  
We analyse the survival probability 
\begin{eqnarray}
p_{\rm sur}(t) = \mathbb{P}\left[ \tilde{T} > t\right], 
\end{eqnarray}
of the first-passage time with one absorbing boundary,
\begin{eqnarray}
\tilde{T} = {\rm inf}\left\{t: S_{\rm tot}(t) = s_0 \right\},
\end{eqnarray} 
 for  continuous stochastic processes $\vec{X}(t)$.    If the set $\left\{t: S_{\rm tot}(t) = s_0 \right\}$ is empty, then  $\tilde{T} =  \infty$.

We  use the fluctuation relation (\ref{stoppingInt}) since  for first-passage times with one absorbing boundary $S_{\rm tot}(t)$ is  unbounded for  $t\in [0,T]$. 
Applying (\ref{stoppingInt}) to $\tilde{T}$ we obtain 
\begin{eqnarray} 
 1 = p_{\rm sur}(t)\langle e^{-S_{\rm tot}(t)} \rangle_{\rm sur}    + (1-p_{\rm sur}(t)) e^{-s_0}
\end{eqnarray}
and thus also relation (\ref{eq:survxxxxxx}), i.e., 
\begin{eqnarray}
p_{\rm sur}(t) = \frac{e^{-s_0}  - 1}{e^{-s_0} -   \langle e^{-S_{\rm tot}(t)} \rangle_{\rm sur}   }, 
\end{eqnarray}   
where 
\begin{eqnarray}
\langle e^{-S_{\rm tot}(t)} \rangle_{\rm sur}  = \langle e^{-S_{\rm tot}(t)} I_{\tilde{T}>t} \rangle/\langle  I_{\tilde{T}>t}\rangle.
\end{eqnarray}
For positive values of $s_0$ we expect that  $\lim_{t\rightarrow \infty}p_{\rm sur}(t) = 0$.   This implies that $\lim_{t\rightarrow \infty}  \langle e^{-S_{\rm tot}(t)} \rangle_{\rm sur} = \infty$.    For negative values of $s_0$ we expect that  $\lim_{t\rightarrow \infty}p_{\rm sur}(t) =  1 - e^{s_0} = \mathbb{P}\left[S_{\rm inf} \geq s_0\right]$ which holds if 
 $\lim_{t\rightarrow \infty}\Big\langle e^{-S_{\rm tot}(t)} I_{\tilde{T}>t} /\langle I_{\tilde{T}>t} \rangle \Big\rangle = 0$.

\subsection{Number of crossings}\label{sec:crossings}
We consider the number of times $N_\times$  entropy production crosses an interval $[-\Delta, \Delta]$ from the negative side to the positive side, i.e., in the direction $-\Delta \rightarrow \Delta$.   We can bound the distribution of $N_{\times}$ using a sequence of stopping times.

The probability that $N_\times>0$  is equal to the probability that the infimum is smaller or equal than $-\Delta$, and therefore using (\ref{eq:cumInf}) we obtain
\begin{equation}
\mathbb{P}(N_\times>0;\Delta) = \mathbb{P}\left[S_{\rm inf} \leq -s_-\right] \le e^{-\Delta}.
\end{equation}
Applying the conditional fluctuation relation (\ref{stoppingInt2})  on two sequences of stopping times, we derive  in the \ref{sec:crossinC}  the inequality
\begin{equation}
\mathbb{P}\left[N_{\times}\geq n+1| N_{\times}\geq n\right] \le e^{-2\Delta}\qquad \mathrm{with}\quad n>0,  \label{eq:formulaToBe}
\end{equation}
and therefore  
\begin{equation}
\mathbb{P}(N_\times>n)\le e^{-\Delta (2n+1)} ,
\end{equation}
which is the inequality (\ref{ineq:crossing}).
The probability of observing a large number of crossing decays  thus at
least exponentially in $N_\times$.
For continuous processes $\vec{X}(t)$ the probability mass function $\mathbb{P}(N_\times)$ is a universal statistic given by
\begin{equation}
\mathbb{P}(N_\times=n)= \left\{
\begin{array}{lr}
1-e^{-\Delta} & n= 0, \\
2\sinh ( \Delta) e^{-2 n\Delta} & n\ge 1,
\end{array}\right.
\end{equation}
which is the same relation as derived in \cite{pigolotti2017generic} for overdamped Langevin processes.

    \section{The influence of finite statistics on the integral fluctuation relation at stopping times }\label{sec:finiteStat}
    In empirical situations we may want  to use the integral fluctuation relation at stopping times to test whether a given process is the stochastic entropy production~\cite{speck2007distribution, koski2013distribution, ciliberto2013heat}.  In these cases we have to deal with finite statistics.  It is useful to know how many experiments are required to verify the integral fluctuation relation at stopping times with  a certain accuracy.   Therefore, in this section we discuss the influence of finite statistics on tests of the integral fluctuation relation at stopping times.  
      
  We consider the  case where 
 $T$ is  a two-boundary first-passage time for the entropy production of a continuous stationary process, i.e.,  $T = T_{\rm fp} = {\rm inf}\left\{t\geq 0: S_{\rm tot}(t) \notin (-s_-,s_+)\right\}$.  We imagine to estimate in 
an experiment the average  $\langle e^{-S_{\rm tot}(T_{\rm fp})} \rangle$ using an empirical average   over  $m_{\rm s}$ realisations, 
  \begin{eqnarray}
 A= \frac{1}{m_{\rm s}}\sum^{m_{\rm s}}_{j=1}e^{-S_{\rm tot}(T_j)} =  e^{-s_+} + \frac{N_-}{m_{\rm s}} \left(e^{s_-}-e^{-s_+}\right),  \label{eq:A}
 \end{eqnarray} 
 where the $T_j$ are the different outcomes of the first passage time $T_{\rm fp}$ and where $N_-$ is the number of trajectories that have terminated at the negative boundary, i.e., for which
 $S_{\rm tot}(T_j) = -s_-$.   The expected value of the sample mean is thus 
  \begin{eqnarray}
 \langle A \rangle = 1 ,
    \end{eqnarray}
and the variance of the sample mean is
  \begin{eqnarray}
\sigma^2_A =  \langle A^2\rangle - \langle A \rangle^2 =  \frac{(1-e^{-s_+})( e^{s_-}-1)}{m_{\rm s}}. \label{eq:sigmaA}
   \end{eqnarray}
Hence, for small enough values of  $s_-$ a few samples ${m_{\rm s}}$ will be  enough to test the stopping-time fluctuation relation.    For large enough $s_-$, we obtain   $\langle A^2\rangle - \langle A \rangle^2  \sim e^{s_-}/m_{\rm s}$.      The number of required  samples ${m_{\rm s}}$  scales exponentially in the  value of the negative threshold  $s_-$.   
The full distribution of the empirical estimate $A$  of $\langle e^{-S_{\rm tot}(T_{\rm fp})}\rangle$ is given by
 \begin{eqnarray}
\fl  p_A(a) = \frac{1}{(e^{s_-}-e^{-s_+})^{m_{\rm s}}}\sum^{m_{\rm s}}_{n=0}\delta_{a,n \:e^{s_-} + (m_{\rm s}-n) e^{-s_+}} \left(\begin{array}{c}m_{\rm s} \\ n\end{array}\right) (1-e^{-s_+})^{n} ( e^{s_-}-1)^{m_{\rm s}-n}, \label{eq:pa}
  \end{eqnarray}
  where we have again used $\delta$ for the Kronecker delta.   
  
Since we know the  full distribution of the empirical average $A$, the integral fluctuation relation $\langle e^{-S_{\rm tot}(T_{\rm fp})}\rangle=1$  can be tested in experiments: given a certain observed value of $A\neq 1$ we can use the distribution (\ref{eq:pa}) to compute its $p$-value, i.e., the probability   to observe a deviation from $1$  larger or equal than the empirically observed  $|A-1|$. 
  
 \section{Examples} \label{Sec:illustration} 
 We illustrate the bounds~(\ref{eq:QT}) and  (\ref{eq:carnotStoppxxxx}) on two simple examples of systems described by Langevin equations.   We demonstrate a randomly stopped process can extract work from a stationary isothermal process and we show that heat engines can surpass the Carnot efficiency at stopping times.    Moreover, the mean of the extracted heat and the    performed work are bounded by the second law of thermodynamics at stopping times.
 
 \subsection{Heat extraction in an isothermal system} 
 \label{sec:langevin} 
 We illustrate the second law~(\ref{eq:QT}), for the average heat at stopping times, on the case of  
  a colloidal  particle on a ring  that is driven by an external force $f$ and moves in a potential  $v(y)$, where $y \in [0,2\pi \ell]$ denotes the position on the ring and where $\ell$ is the radius of the ring.   We ask  how much heat $\langle Q(T^\ast)\rangle$ a colloidal particle can extract on average from its environment at the  time $T^\ast$  when the particle reaches for the first time  the highest peak of the  landscape.

 \begin{figure}[t]
\centering
 \hspace{-0.5cm}
 \subfigure[Plot of the potential~(\ref{eq:energy}). ]
{\includegraphics[width=0.5\textwidth]{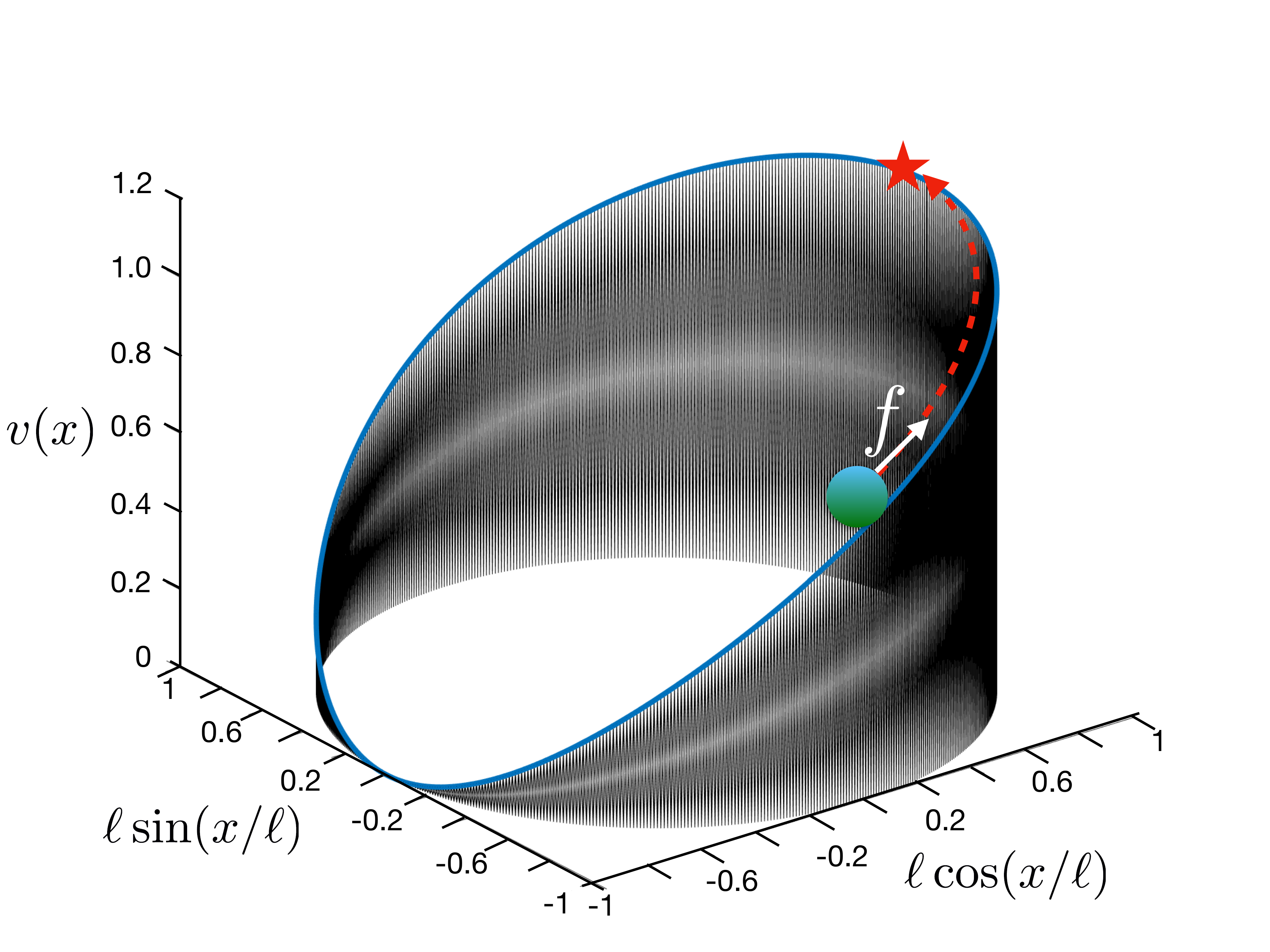}\label{fig1a}} 
\subfigure[Ilustration of the bound  (\ref{eq:QT})  for the absorbed heat at stopping times.]
{\includegraphics[width=0.5\textwidth]{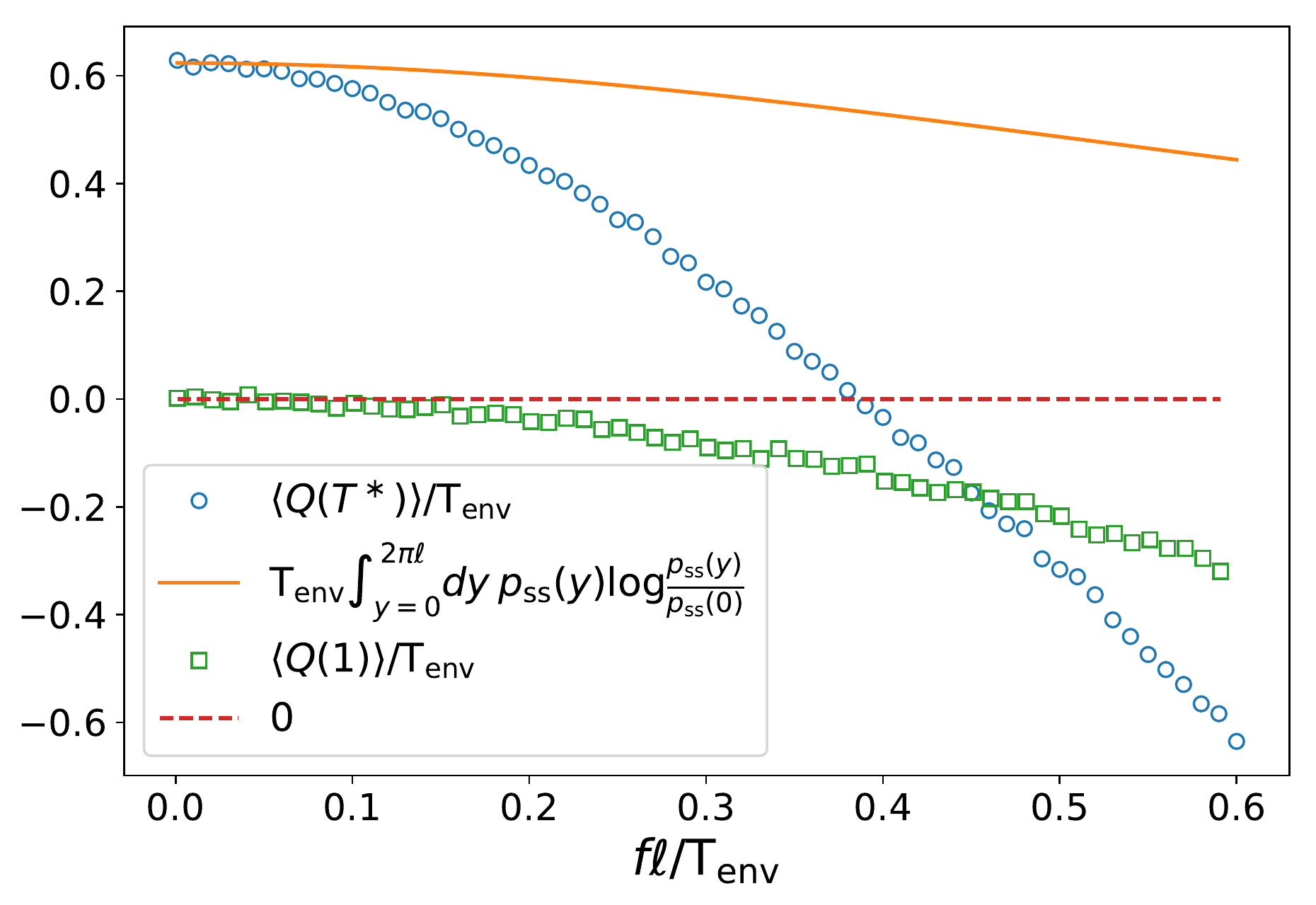}\label{fig1b}} 
\caption{Heat extraction  of a colloidal particle in a nonequilibrium stationary state.
Panel (a):  Illustration of a colloidal particle that is driven by a nonconservative $f$ until the time $T^\ast$ when it reaches the highest point of the 'hill' denoted by the star, at which the process is stopped.   The plotted energy function $v(x)$ is given by (\ref{eq:energy})  with  parameters  $\mathsf{T}_{\rm env} = 1$ and $\ell = 1$. Panel (b): Illustration of the bound~(\ref{eq:QT})  for the absorbed heat $\langle Q(T^\ast)\rangle$ at stopping times $T^\ast$   in the   model~(\ref{eq:langevin1})  with the energy function   plotted in panel (a).    The parameters used are $\mu = 1$, $\mathsf{T}_{\rm env} = 1$, $\ell = 1$.    Markers denote empirical averages  that estimate  $\langle Q(T^\ast)\rangle$ (blue circles) and $\langle Q(1)\rangle$ (green squares) using  $10000$ simulated trajectories. The solid orange line denotes the bound  on the right-hand side of~(\ref{eq:QBoundxx}) which follows from the second law of thermodynamics~(\ref{eq:strong2}) at stopping times, and the red dashed line is simply equal to zero.    } \label{fig1}
\end{figure}   

   We assume that the dynamics of the colloidal particle  is governed by an overdamped Langevin equation of the form 
  \begin{eqnarray}
  \frac{{\rm d}X}{{\rm d} t} = -\mu \frac{\partial v(X(t))}{\partial x}  + \mu f +  \sqrt{2d} \zeta(t) \label{eq:langevin1}
    \end{eqnarray}   
    where $X(t)\in\mathbb{R}$ ,  $\mu$ is the mobility coefficient, $d$ is the diffusion coefficient, and $\zeta(t)$ is a Gaussian white noise with $\langle \zeta(t)\zeta(t')\rangle = \delta(t'-t)$.       We assume that the environment surrounding the particle is in equilibrium at a temperature  $\mathsf{T}_{\rm env}$, so that  Einstein relation $d = \mathsf{T}_{\rm env} \mu$ holds.   The potential $v(x)$ is a periodic function with period $2\pi\ell$, i.e., $v(x) = v(x+2\pi \ell)$.  The actual position of the particle is  given by  the variable $Y(t) = X(t) - N(t)2\pi\ell \in [0,2\pi \ell)$, where $N(t)\in \mathbb{Z}$
 is the winding number, i.e., the net number of times the particle has traversed the ring.   
The heat $Q$ can be expressed as~\cite{sekimoto1998langevin, sekimoto2010stochastic} 
\begin{eqnarray}
Q(t) = v(X(t))-v(X(0)) + f \int^t_0 {\rm d}X(t'),
\end{eqnarray}
and the stationary distribution is  \cite{reimann2002brownian} 
\begin{eqnarray}
p_{\rm ss}(y) \sim e^{-\frac{v(y)-fy}{\mathsf{T}_{\rm env}}} \int^{y+2\pi\ell}_{y}dy' \:e^{\frac{v(y')-fy'}{\mathsf{T}_{\rm env}}} . \label{eq:pss}
\end{eqnarray}

We simulate the model (\ref{eq:langevin1}) for a  periodic potential  of the form \cite{PhysRevE.74.061113},
      \begin{eqnarray}
    v(x) = \mathsf{T}_{\rm env} \ln \left(\cos (x/\ell) + 2 \right), \label{eq:energy}
    \end{eqnarray}    
     which is illustrated in figure~\ref{fig1a} for $\ell=1$ and $ \mathsf{T}_{\rm env}=1$.
      In this case,  the time $T^\ast$  when the particle reaches for the first time  the highest peak of the  landscape  is
      \begin{eqnarray}
T^\ast = {\inf}\left\{ t\geq 0:  X(t) = 0 \right\},\label{eq:minTx}
\end{eqnarray}    
and the stationary distribution of $X(t)$ is given by 
\cite{PhysRevE.74.061113}
\begin{eqnarray}
p_{\rm ss}(x) =\frac{3 \left( g^2\left(2+\cos(x/\ell)\right)  -   g\sin \left(x/\ell\right)+2\right)}{2\pi \ell \left(3 g^2+2  \sqrt{3}\right) \left(\cos \left(x/\ell\right)+2\right)}.
\end{eqnarray}
   Hence, the bound~(\ref{eq:langevin1}) reads
   \begin{eqnarray}
 \fl   \langle Q(T^\ast)\rangle \geq \mathsf{T}_{\rm env} \int {\rm d} x\: p_{\rm ss}(x) \:\log p_{\rm ss}(x)  - \mathsf{T}_{\rm env} \log  \frac{3\: g^2  +2}{2\pi \ell \left(3 g^2+2  \sqrt{3}\right) }. \label{eq:QBoundxx}
   \end{eqnarray}
    
    In figure~\ref{fig1b} we illustrate the bound  (\ref{eq:QBoundxx}) for  $\langle Q(T^\ast)\rangle$ as a function of the nonequilibrium driving $f$.  We find that heat absorption at stopping times is significant at small values of $f$ and in the linear-response limit of small $f$ the bound (\ref{eq:QBoundxx}) is tight; the tightness of the bound  (\ref{eq:langevin1}) holds in general for  recurrent Markov processes in the linear response limit.       For comparison we also plot the mean heat dissipation  $\langle Q(t)\rangle$  at a fixed time $t=1$.   While  $\langle Q(T^\ast)\rangle$ is positive at small values of $f$, the dissipation $\langle Q(t)\rangle$ at fixed times is always negative.      Note that at intermediate times  $\langle Q(1)\rangle >  \langle Q(T^\ast) \rangle$, but if  we  increase $f$ furthermore then eventually $\langle Q(1)\rangle <  \langle Q(T^\ast) \rangle$ (not shown in the figure).

 \subsection{Illustration of an empirical test of the integral fluctuation relation at stopping times} 

We  use the  the model described in   Subsection~\ref{sec:langevin} to illustrate how the  integral fluctuation relation at stopping times~(\ref{eq:strong}) with  $T = T_{\rm fp} = {\rm inf}\left\{t\geq 0: S_{\rm tot}(t) \notin (-s_-,s_+)\right\}$
 can be tested in an experimental setup.        To this aim, we will verify whether the quantity $A$  defined in ~(\ref{eq:A}), which is the empirical average of $e^{-S_{\rm tot}(T)}$, converges for $m_s\rightarrow \infty$ to one.  
 We use the results  on the statistics of the empirical average of  $A$  described in Section~\ref{sec:finiteStat} to validate  the statistical significance of the experimental results.  
 
In Figure~\ref{fig3} we plot the empirical average~(\ref{eq:A}) as a function of the number of samples~$m_s$ for  ten  simulation runs.        We also plot the theoretical curves $1+\sigma_A$ and $1-\sigma_A$, with $\sigma_A$ the standard deviation of $A$, see (\ref{eq:sigmaA}).     We observe in Figure~\ref{fig3}  that all test runs lie within the $1\pm\sigma_A$  confidence intervals, and we can thus  conclude that numerical experiments are in agreement with  the integral fluctuation relation at stopping times and thus  also with the  fact that   $e^{-S_{\rm tot}(t)}$ is martingale.

   \begin{figure}[t]
\centering
 \hspace{-0.5cm}
{\includegraphics[width=0.5\textwidth]{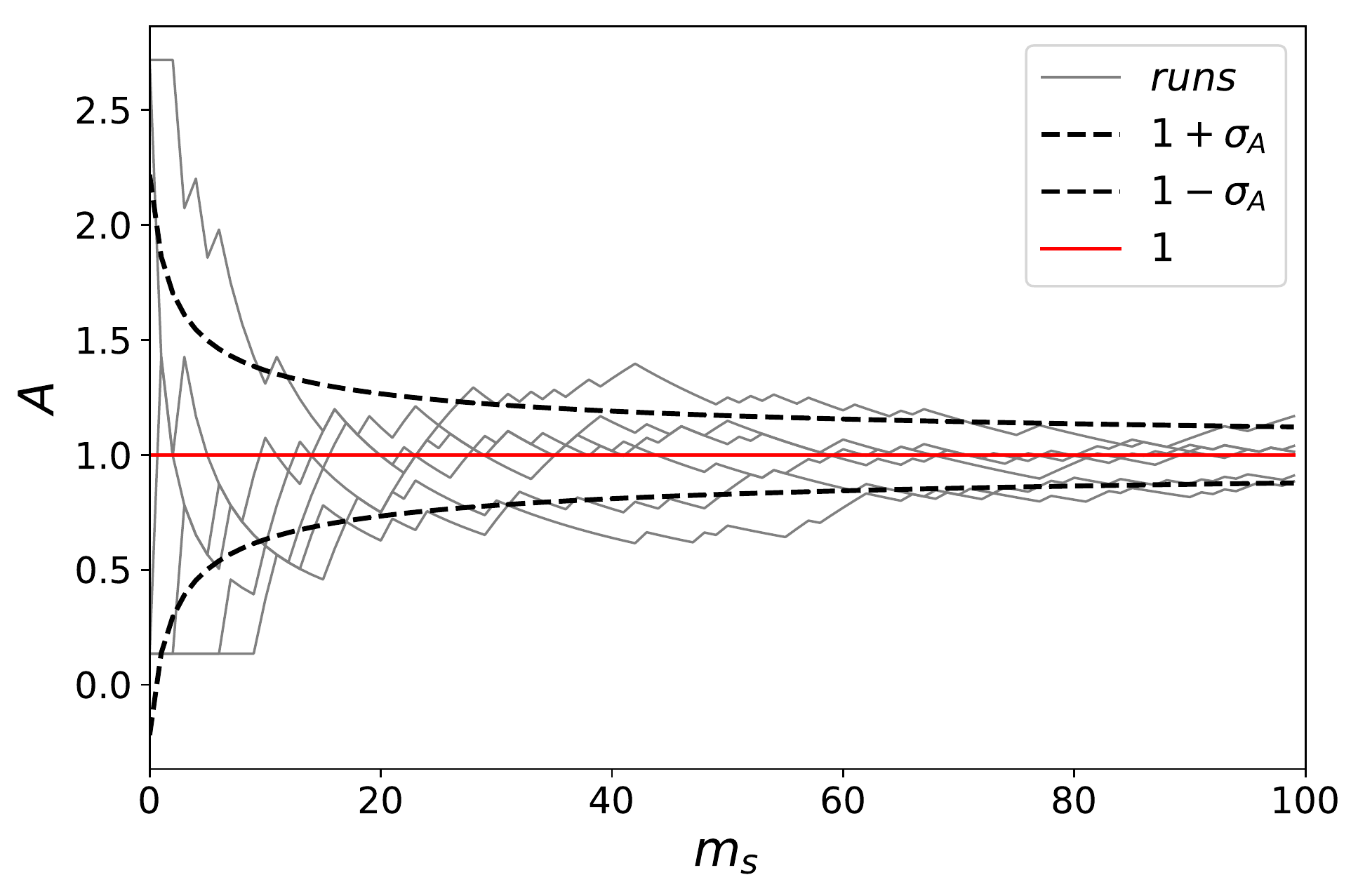}} 
\caption{ Empirical test of the  integral fluctuation relation at stopping times  for  $T = T_{\rm fp}  = {\rm inf}\left\{t\geq 0: S_{\rm tot}(t) \notin (-s_-,s_+)\right\}$.  We plot the empirical average $A$, see (\ref{eq:A}),  for ten simulation runs as a function of the number of samples $m_s$.   The simulation runs are  for the model defined in Section~\ref{sec:langevin}  with parameters  $\mu = \mathsf{T}_{\rm env} = \ell = 1$  and $f = 0.1$, as in Figure~\ref{fig1}.   The threshold parameters that define the stopping time are 
  $s_+ = 2$ and  $s_- = 1$.  Dashed lines denote the $1\pm\sigma_A$  confidence intervals using formula~(\ref{eq:sigmaA}).   The red line denotes $A=1$.   \label{fig3}}
\end{figure}  
 
 \subsection{Super Carnot efficiency  for heat engines at stopping times} 
  \label{sec:engines}
  We illustrate the bound~(\ref{eq:carnotStoppxxxx}) for the efficiency  of stationary stochastic heat engines at stopping times with  
  a Brownian gyrator~\cite{filliger2007brownian}, which is arguably one of the simplest models of a Feynman ratchet. This system  is described by two degrees of freedom $x_1,x_2$  that are driven by an external force field $f(x_1,x_2)$ and interact via a potential  $v(x_1,x_2)$.   Two thermal reservoirs at temperatures $\mathsf{T}_{\rm h}$ and $\mathsf{T}_{\rm c}$, with $\mathsf{T}_{\rm h}>\mathsf{T}_{\rm c}$, interact independently with the coordinates $x_1$ and $x_2$ of the system, respectively.  We are specifically interested in the efficiency $\eta_{T}=-\langle W(T) \rangle /\langle Q_{\rm h}(T) \rangle$, which is   the ratio between the work $-W(T)$  the gyrator performs on its suroundings  in a time interval $[0,T]$, and the heat $Q_{\rm h}(T)$ absorbed by the gyrator  in the same time interval,   with $T$ the  stopping time at which a specific  criterion is first satisfied.   The efficiency $\eta_{T}$ is a measure of the average amount of work the gyrator performs on  its environment.
  

 We consider a Brownian gyrator described by the two coupled stochastic differential equations~\cite{filliger2007brownian,argun2017experimental,pietzonka2018universal,cerasoli2018asymmetry,manikandan2019efficiency}
   \begin{eqnarray}
  \frac{{\rm d} X_1}{{\rm d} t} &=& -\mu \frac{\partial v(X_1(t),X_2(t))}{\partial x_1}  + \mu  \: f_1 (X_1(t),X_2(t))+  \sqrt{2d_1} \zeta_1(t), \label{eq:gyrator1} \\
  \frac{ {\rm d} X_2}{ {\rm d} t} &=& -\mu \frac{\partial v(X_1(t),X_2(t))}{\partial x_2}  + \mu \: f_2 (X_1(t),X_2(t))+  \sqrt{2d_2} \zeta_2(t) \label{eq:gyrator2}.
    \end{eqnarray}  
    Here $\mu$ is the mobility coefficient, $d_1=\mu\mathsf{T}_{\rm h}$  and $d_2=\mu\mathsf{T}_{\rm c}$ are the diffusion coefficients of the two degrees of freedom, $v$ is a potential,  and $f_1$ and $f_2$ are two external nonconservative forces, whose functional form we specify below.   We use the model from Ref.~\cite{pietzonka2018universal}, for which the potential is 
    \begin{equation}
    v(x_1,x_2)=\frac{1}{2}\left( u_1 x_1^2 + u_2 x_2^2 + c x_1 x_2  \right),
    \label{eq:potgyrator}
    \end{equation}
    with $u_1,u_2,c>0$ and  $c <\sqrt{u_1 u_2}$ and for which the two components of the external nonconservative force are
    \begin{equation}
    f_1 (x_1,x_2) = k x_2,\quad f_2 (x_1,x_2) = -kx_1. \label{eq:forcesgyrator}
    \end{equation}
The two thermal reservoirs induce two  stochastic forces  with amplitudes $d_1$ and $d_2$, which appear in Eqs.(\ref{eq:gyrator1}-\ref{eq:gyrator2}) as   two independent Gaussian white noises   $\zeta_1$  and $\zeta_2$  with zero mean and autocorrelation ($i,j=1,2$)
\begin{equation}
\langle \zeta_1 (t)\rangle = \langle \zeta_2 (t)\rangle=0,\quad \langle \zeta_i (t)\zeta_j (t')\rangle = \delta_{i,j} \delta (t-t').
\end{equation}

 \begin{figure}[t]
\centering
\includegraphics[width=0.5\textwidth]{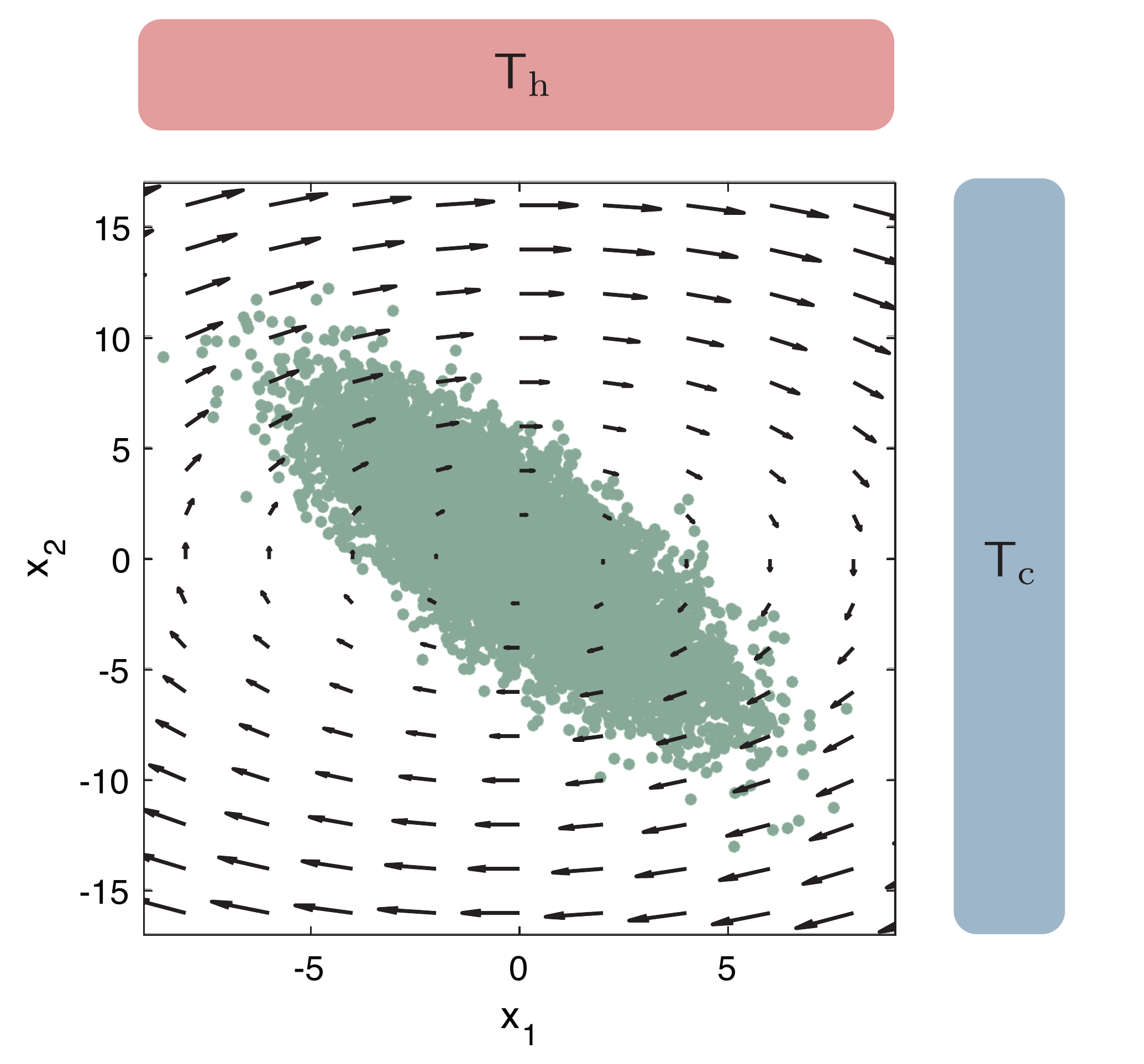}
\caption{Snapshots of a Brownian gyrator (green circles), described by Eqs.~(\ref{eq:gyrator1}) and~(\ref{eq:gyrator2}), sampled from the nonequilibrium stationary distribution. The black arrows illustrate the non-conservative forces given by Eqs.~(\ref{eq:forcesgyrator}). The system is trapped with a potential given by Eq.~(\ref{eq:potgyrator}) and put simultaneously in contact with a hot (red box) and a cold (blue box) reservoir that act on the $x_1$ and $x_2$ coordinates,  respectively. Values for the parameters are: $\mu=1$, $u_1=1$, $u_2=1.2$, $\mathsf{T}_{\rm c}=1$, $\mathsf{T}_{\rm h}=7$, $c=0.9$, $k=0.23$. The markers (green circles) are obtained from the $(x_1(t),x_2(t))$ coordinates of $10^4$ numerical simulations of Eqs.~(\ref{eq:gyrator1}) and~(\ref{eq:gyrator2}) evaluated at $t=5$.  }\label{figG1}
\end{figure}   

 Because of the external driving forces $f_1$ and $f_2$ and the presence of two thermal reservoirs at different temperatures, the gyrator develops  a nonequilibrium stationary state characterised by a current in the clockwise direction, see Fig.~\ref{figG1}, and a non-zero entropy production. At stationarity, we measure the work  $W$ that the  external driving force  exerts on the gyrator    and the net heat $Q_{\rm h}$ and $Q_{\rm c}$ that  the system absorbs from  the hot and cold reservoirs, respectively.   Following Sekimoto~\cite{sekimoto1997kinetic,sekimoto1998langevin}, these quantities are, respectively, 
  \begin{eqnarray}
\hspace{-1cm}  W(t) &=&  \int_0^t  f_1 (X_1(t'),X_2(t')) \circ {\rm d}X_1(t')  +   \int_0^t  f_2 (X_1(t'),X_2(t')) \circ {\rm d}X_2(t') ,\label{eq:WSeki}\\
\hspace{-1cm}  Q_{\rm h}(t) &=&  \int_0^t  \frac{\partial v(X_1(t'),X_2(t'))}{\partial x_1} \circ {\rm d}X_1(t')   - \int^t_0  f_1 (X_1(t'),X_2(t'))  \circ {\rm d}X_1(t'),\label{eq:QhSeki}\\
\hspace{-1cm}    Q_{\rm c}(t) &=&  \int_0^t  \frac{\partial v(X_1(t'),X_2(t'))}{\partial x_2} \circ {\rm d}X_2(t')  - \int^t_0   f_2(X_1(t'),X_2(t')) \circ {\rm d}X_2(t'),
 \end{eqnarray}
 where $\circ$ denotes that the stochastic integrals are interpreted in the Stratonovich sense.  When $k<k_{\rm s}$  this system operates as an engine  \cite{pietzonka2018universal},  i.e., $\langle  W(t) \rangle<0$, $\langle  Q_{\rm h}(t) \rangle>0$ and $\langle  Q_{\rm c}(t) \rangle>0$, with 
\begin{equation}
k_{\rm s} := c \frac{ \eta_{\rm C} }{2-\eta_{\rm C}},
\end{equation} the stall parameter and $\eta_{\rm C}$ Carnot's efficiency~(\ref{eq:Carnotefficiency}).     The   efficiency of the  engine in the nonequilibrium stationary state satisfies
\begin{equation}
  \eta = - \frac{\left\langle \frac{ {\rm d}W}{{\rm d}t}\right\rangle}{\left\langle \frac{{\rm d}Q_{\rm h}}{{\rm d}t}\right\rangle}=\frac{2k}{c+k}\leq \eta_{\rm C}. \label{eq:etaTh}
\end{equation} 
Note that when $k\rightarrow k_{\rm s}$, then $\eta\rightarrow \eta_{\rm C}$.

 \begin{figure}[t]
\centering
\hspace{-0.5cm}
\includegraphics[width=0.4\textwidth]{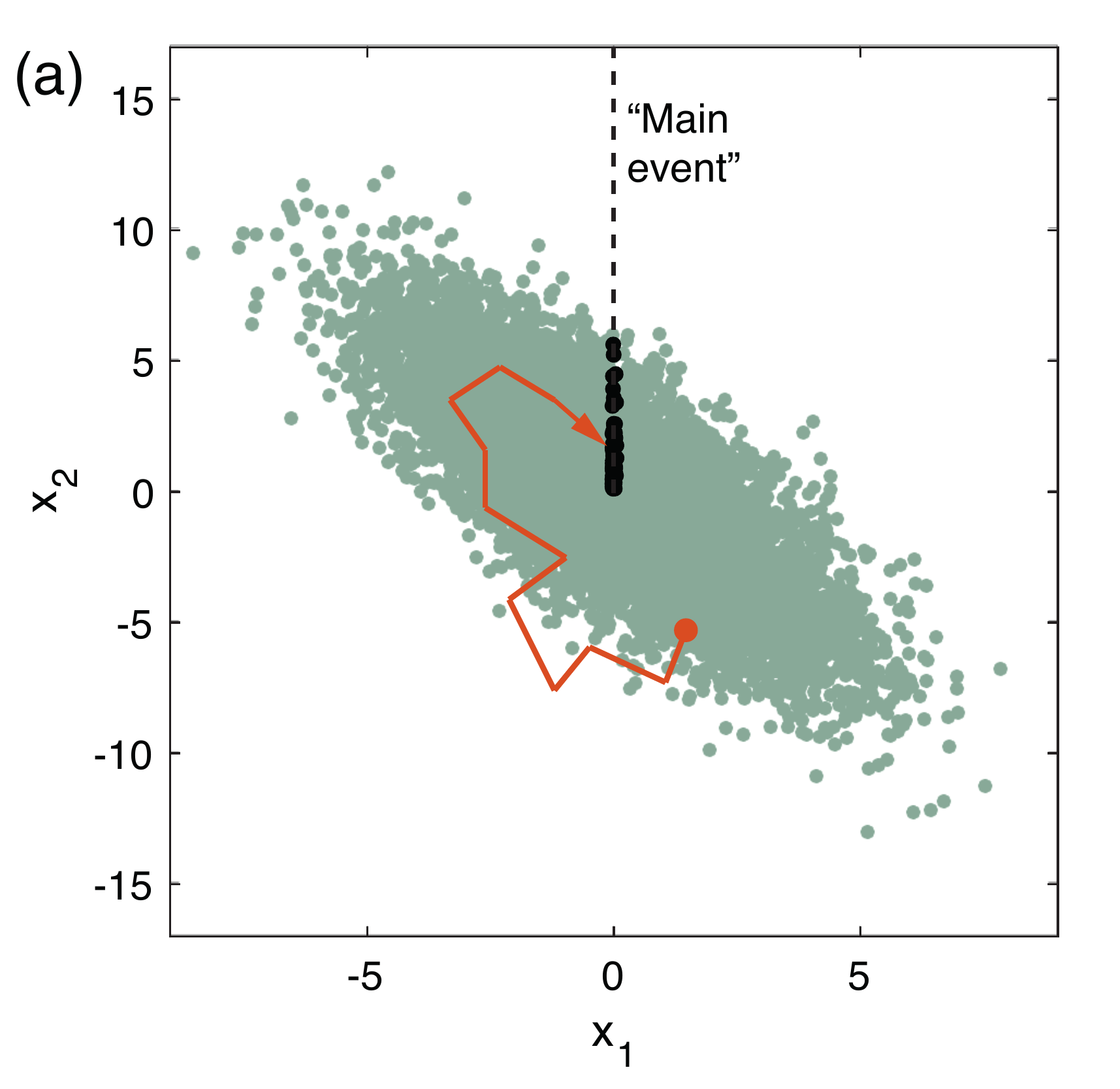}
\includegraphics[width=0.47\textwidth]{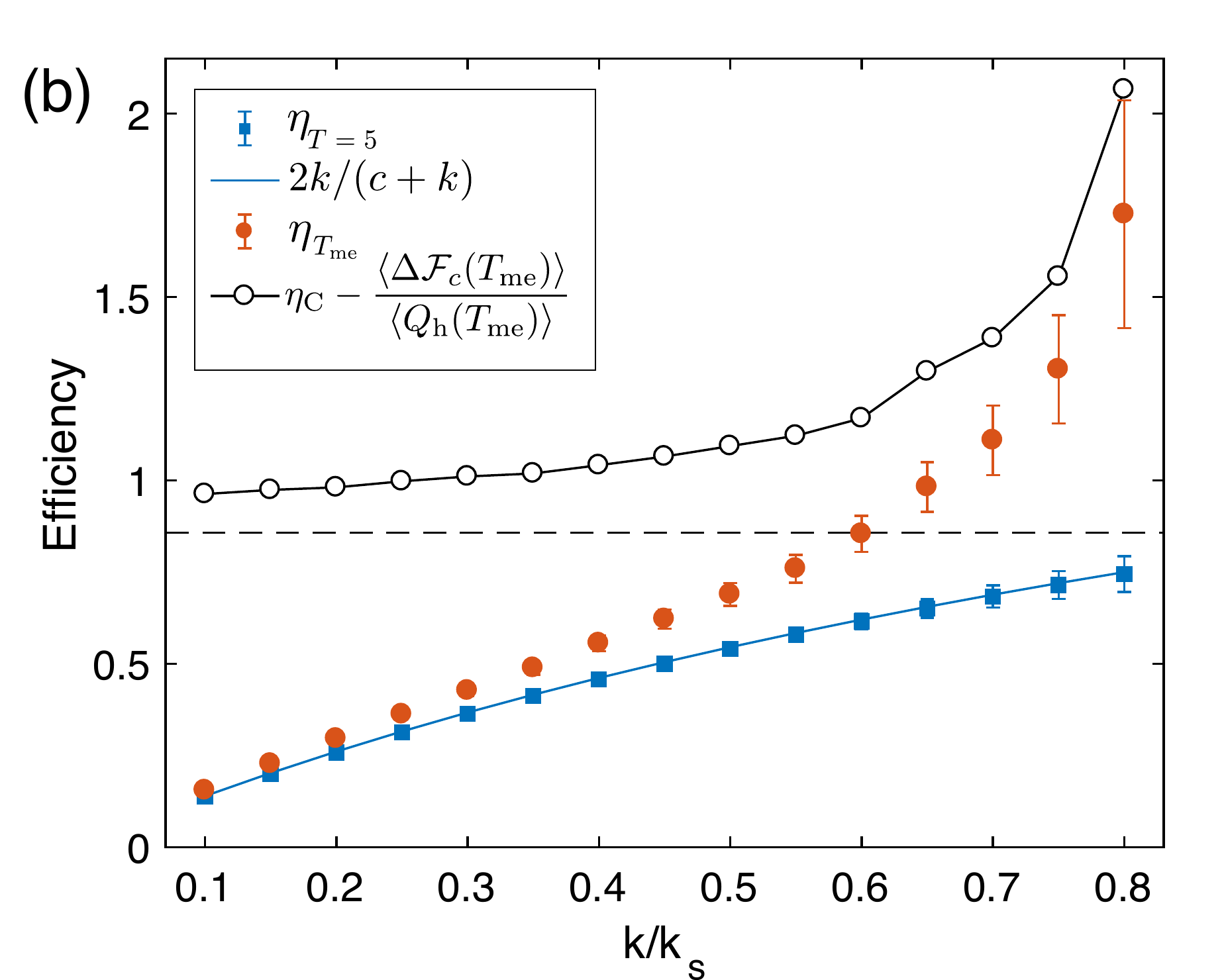}
\caption{Efficiency at stopping times for the Brownian gyrator. (a) Illustration of the stopping time at the main event with stopping time $T_{\rm me}$. A gyrator drawn initially from its stationary distribution (red circle) is monitored until its location crosses the black dashed line in the clockwise direction (red arrow). The location of $100$ gyrators at the stopping time are shown with black filled circles whereas the green circles denote the location of the gyrator at time $t=5$.  (b) Simulation results for the efficiency  $\eta_t =-\langle W(t)\rangle/\langle Q_{\rm h}(t)\rangle$ at fixed time $t=5$ (blue squares) and  the efficiency $\eta_{T_{\rm me}} =-\langle W(T_{\rm me})\rangle/\langle Q_{\rm h}(T_{\rm me})\rangle$ at the main event  $T_{\rm me}$  (red circles), defined by Eq.~(\ref{eq:Tme}),   as a function of the parameter $k/k_{\rm s}$ that quantifies the strength of the nonequilibrium driving, see (\ref{eq:forcesgyrator}).  The blue line is the theoretical value of the  efficiency  at large times given by Eq.~(\ref{eq:etaTh}),  and the black circles correspond to  the bound~(\ref{eq:carnotStoppxxxx}); the black line is a guide to the eye. The horizontal black dashed line is set at Carnot efficiency. Values of the parameters used in simulations are $\mu=1$, $u_1=1$, $u_2=1.2$, $\mathsf{T}_{\rm c}=1$, $\mathsf{T}_{\rm h}=7$, and $c=0.9$;  markers denote average values   estimated from $10^4$ independent realisations initially sampled  from the stationary state. Numerical simulations are performed with the   Heun's numerical integration scheme with a time step $\Delta t=10^{-3}$ \cite{sekimoto2010stochastic}. Error bars  denote the  standard errors of the empirical mean. }\label{figG2}
\end{figure}
We now investigate the efficiency of the Brownian gyrator at stopping times, $\eta_{T}=-\langle W(T)\rangle /\langle Q_{\rm h}(T)\rangle$. The simplest example of  stopping times are the  trivial stopping times  $T=t$, with $t$ a fixed time, for which $\eta_{ t}\leq \eta_{\rm C}$.   A more interesting example is the time $T_{\rm me}$ of the   {\it main event}, i.e. the gyrator crosses the positive $x_2$ axis while moving in the clockwise direction, occurs for the first time. Mathematically, $T_{\rm me}$ can be defined as follows: let $z(t)=x_1(t)+i x_2(t)$ be a complex number whose real and imaginary parts are  $x_1(t)$ and $x_2(t)$ respectively; we define $z(t)=r(t)e^{i\varphi(t)}$, with $r(t)=\sqrt{x_1^2(t) + x_2^2(t)}$ its modulus and  with $\varphi(t) = \tan^{-1} (x_2(t)/x_1(t))\in [-\pi , \pi]$ its phase; the stopping time at the main event is defined as
\begin{eqnarray}
T_{\rm me} = {\rm inf}\left\{t>0: \lim_{\epsilon\downarrow 0}  \varphi(t - \epsilon) > \pi/2 , \; \lim_{\epsilon\downarrow 0} \varphi(t+\epsilon) \leq  \pi/2 , \;  \varphi(t)>0   \right\}.
\label{eq:Tme}
\end{eqnarray} 
Figure~\ref{figG2}(a) illustrates the stopping  strategy defined by Eq.~(\ref{eq:Tme}) with numerical simulations.
 In Fig.~\ref{figG2}(b), we compare the values of the  efficiency $\eta_t$ at a fixed time  $t=5$   with the efficiency $\eta_{T_{\rm me}} $ at the main event, both as a function of the driving parameter $k$.    We observe that $\eta_t$ is well described by Eq.~(\ref{eq:etaTh}) and thus is smaller than the Carnot efficiency, whereas $\eta_{T_{\rm me}} $  can surpass  the Carnot efficiency  if  the strength of the driving force $k$  is large enough (Fig.~\ref{figG2}(b) red circles).     
   Interestingly, the observed super Carnot efficiencies at stopping times are  in agreement with  the bound~(\ref{eq:carnotStoppxxxx}), $\eta_{T_{\rm me}}\   \leq    \eta_{\rm C} + \langle \Delta \mathcal{F}_c(T_{\rm me})\rangle/\langle Q_{\rm h}(T_{\rm me})\rangle$ and thus compatible with the second law at stopping times~(\ref{eq:strong2}). Moreover, the bound~(\ref{eq:carnotStoppxxxx}) becomes tight when $k$ is large, which corresponds to a close-to-equilibrium limit [cf. Fig.~\ref{fig1b}]. 
   Hence,   efficiencies    of stopped engines can surpass the Carnot bound if   $\langle Q_{\rm h}(T)\rangle >0$ and $\langle\Delta S_{\rm sys}(T)\rangle > \mathsf{T}_{\rm c}\langle \Delta v(T) \rangle $,  which is  consistent with the bound~(\ref{eq:carnotStoppxxxx}).
     
Notice that efficiencies of stopped heat engines can be larger than one because  internal system energy can be converted into useful work.   
To better understand this feature, we plot in Fig.~\ref{figF} the average energetic fluxes at stopping times, namely, $\langle W(T)\rangle$, $\langle Q_{\rm h}(T)\rangle$ and $\langle Q_{\rm c}(T)\rangle $, together with the change  of the internal system energy $\langle \Delta v (T)\rangle$ and the internal system entropy change $\langle \Delta S_{\rm sys} (T)\rangle$. At fixed times $T=  5$, we observe the well-known features of a cyclic heat engine for which $\langle \Delta v (t)\rangle\simeq \langle \Delta S_{\rm sys} (t)\rangle \simeq 0$ (Fig.~\ref{figF}(a)).   Although  the heat and work fluxes of stopped engines have the same sign as  those of   cyclic heat engines, the   $\langle \Delta v (T_{\rm me})\rangle<0 $ and  $\langle \Delta S_{\rm sys} (T_{\rm me})\rangle < 0$, which indicates that on  average  the energy and  entropy of the gyrator  are  at $T_{\rm me}$ smaller than  their   initial values  (Fig.~\ref{figF}(b)). 
 \begin{figure}[t]
\centering
\includegraphics[width=0.9\textwidth]{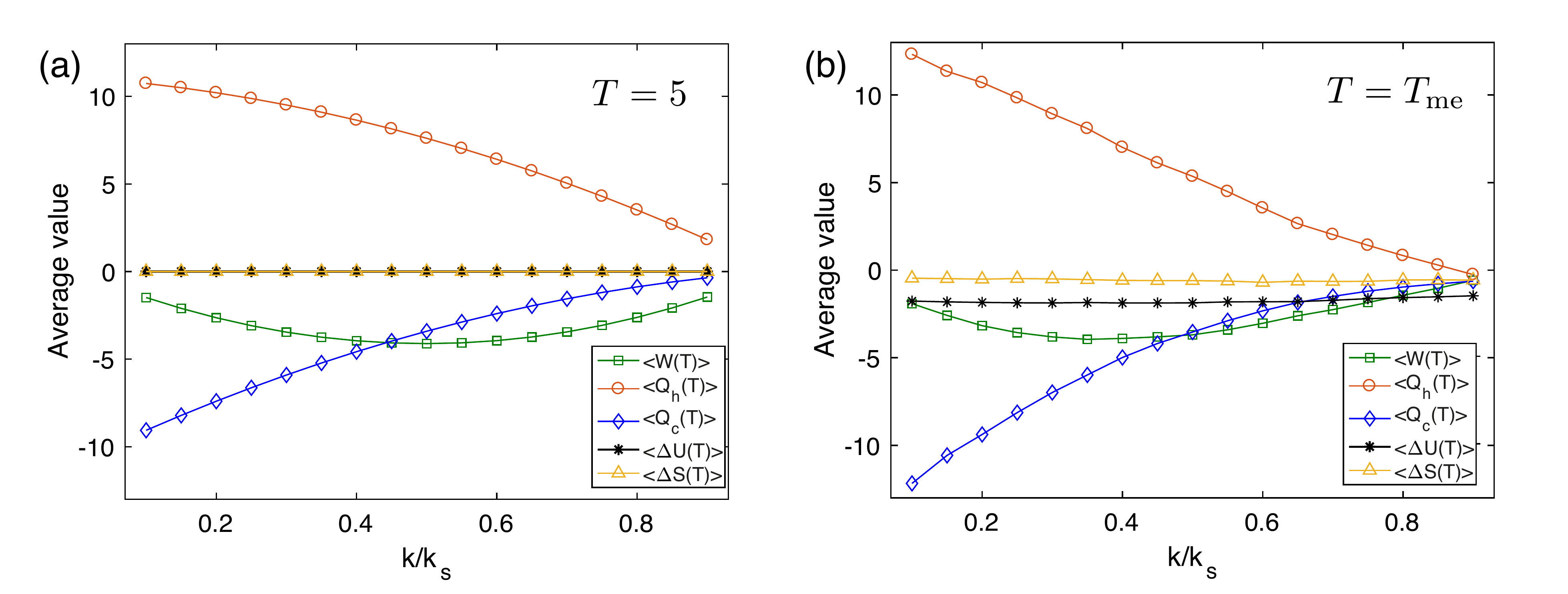}
\caption{Average values of thermodynamic observables (see legend)  in the Brownian gyrator are compared for a fixed time  $T=5$ (a) and the stopping time $T=T_{\rm me}$ (b) defined in  Eq.~(\ref{eq:Tme}).      The values of the parameters used in the numerical simulations are the same as in    Fig.~\ref{figG1}. }\label{figF}
\end{figure}   
This result suggests a recipe in the quest of super Carnot efficiencies at stopping times, namely by designing stopping strategies that lead to a reduction of the energy of the system. In our example, $ \langle \Delta S_{\rm sys} (T_{\rm me})\rangle > \langle \Delta v (T_{\rm me})\rangle$, which enables the appearance of super-Carnot stopping-time efficiencies which are nevertheless compatible with the second law at stopping times.  This  result motivates further research on   so-called type-II efficiencies at stopping times defined as the ratio between the average input and output fluxes of  entropy production~\cite{bejan2016advanced,verley2014universal,raz2016geometric}.

In Fig.~\ref{figH} we illustrate the bound provided by Eq.~(\ref{eq:carnotStoppxxxx}) for  a wider range of parameters.  We observe that when $k$ exceeds the stall parameter $k_{\rm s}$ the thermodynamic fluxes obey $\langle W(T_{\rm me}) \rangle<0$, $\langle Q_{\rm h}(T_{\rm me})\rangle<0$ and  $\langle Q_{\rm c}(T_{\rm me})\rangle<0$.  This behaviour is still compatible with the second law of thermodynamics at stopping times.   Indeed, for  this range of parameters $\eta_{T_{\rm me}}<0$ and the bound~(\ref{eq:carnotStoppxxxx}) becomes $\eta_{T_{\rm me}}\    \geq    \eta_{\rm C} -  \frac{\langle \Delta \mathcal{F}_c(T)\rangle}{\langle Q_{\rm h}(T)\rangle }$.   This bound is corroborated by   numerical simulations in Fig.~\ref{figH}(a) and Fig.~\ref{figH}(b).   We define this type of operation  as a  "type III heater", which is different than the type I and type II heaters~\cite{campisi2015nonequilibrium,rana2016anomalous}, that   appear in  stochastic thermodynamics at fixed times.  
 \begin{figure}[t]
\centering
\includegraphics[width=0.9\textwidth]{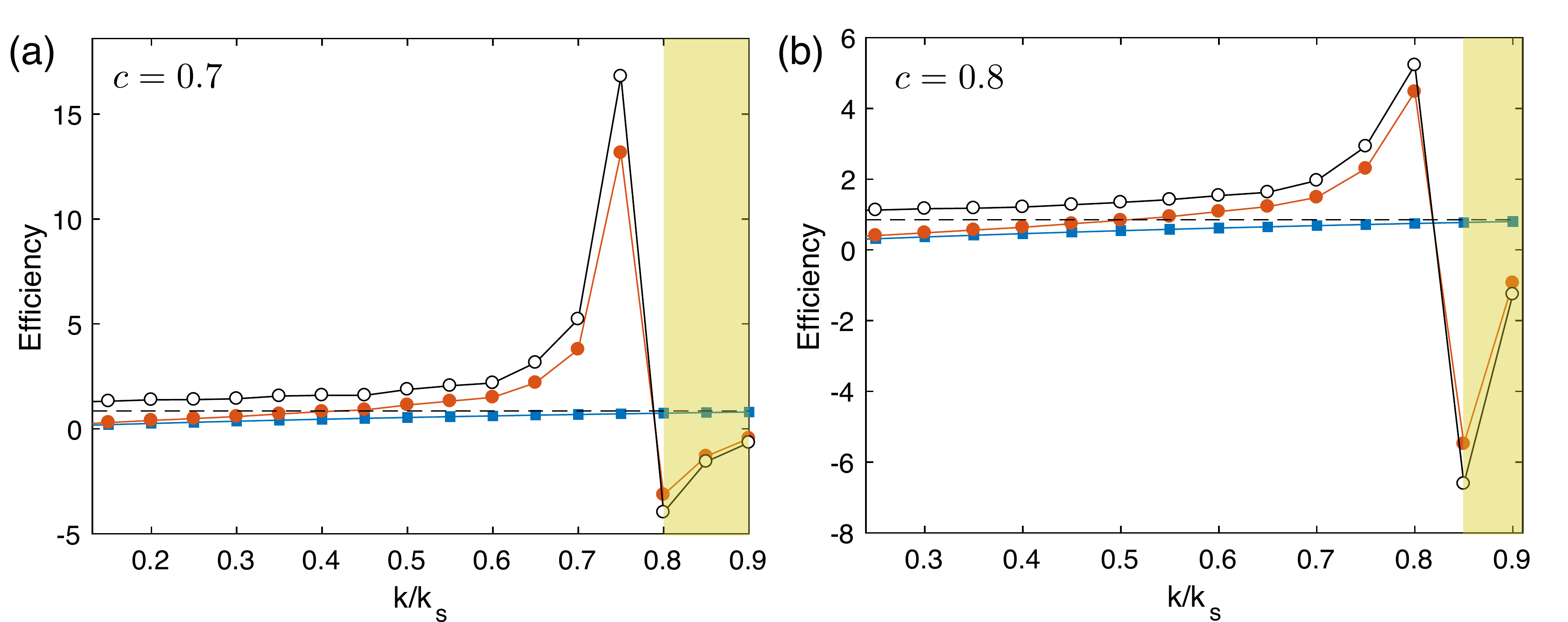}
\caption{ We compare simulation results for the Brownian gyrator of stopping-time efficiencies $\eta_T$ at   fixed time $T=5$  (blue squares)  and at the stopping time of the main event $T=T_{\rm me}$ (red circles). Panels (a) and (b) are obtained for two different values of the parameter $c$ (see legend), with the other simulation parameters set to the same values as in Fig.~\ref{figG1}. The black circles correspond to  the right-hand side of Eq.~(\ref{eq:carnotStoppxxxx}) and the black dashed line is set to Carnot efficiency; the lines between symbols are a guide to the eye. The yellow shaded area illustrates the range of parameters at which the system behaves as a "type-III heater" at stopping times~\cite{campisi2015nonequilibrium,rana2016anomalous}, for which the right-hand side of~(\ref{eq:carnotStoppxxxx}) becomes a lower bound to the stopping-time efficiency.}\label{figH}
\end{figure}

	 \section{Discussion}\label{sec:disc}
We have derived  fluctuation relations at stopping times for the entropy production of  stationary processes.     These fluctuation relations imply a second law of thermodynamics at stopping times --- which states that on average entropy production in a nonequilibrium stationary state always increases, even when the experimental observer measures the entropy production at a stopping time---  and imply that certain fluctuation properties of entropy production are universal.

We have shown that the second law of thermodynamics at stopping times has important consequences for the nonequilibrium thermodynamics of small systems. For instance, the second law at stopping times implies that it is possible to extract on average heat from an isothermal environment by applying    stopping strategies to a physical system in a  nonequilibrium stationary state; heat is thus extracted from the environment without  using a feedback control, as is the case with  Maxwell demons~\cite{toyabe2010experimental, koski2014experimental, parrondo2015thermodynamics}. Furthermore, we have demonstrated, using   numerical simulations with a Brownian gyrator, that the average efficiency of a stationary stochastic heat engine can surpass Carnot's efficiency when the engine is stopped at a cleverly chosen moment.     This result is  compatible with the second law at stopping times, which provides a bound on the  efficiency of stochastically stopped engines.      Note that the heat engines described in this paper are non cyclic devices since they are stopped at the stopping time $T$.    It would be interesting to explore how stochastically stopped engines can be implemented in experimental systems such as, electrical circuits~\cite{ciliberto2013statistical}, autonomous single-electron heat engines~\cite{koski2015chip}, feedback traps~\cite{jun2012virtual}, colloidal heat engines~\cite{martinez2017colloidal},  and  Brownian gyrators~\cite{argun2017experimental}.

Integral fluctuation relations at stopping times  imply   bounds on  the probability of events of negative entropy production that are stronger than those obtained with the integral fluctuation relation at fixed times.   For example, the integral fluctuation relation at stopping times implies that  the cumulative distribution of infima of entropy production is bounded by an exponential distribution with mean $-1$.     Moreover, for continuous processes the integral fluctuation relation at stopping times implies that the cumulative distribution of the global infimum of entropy production is equal to an exponential distribution.   
   A reason why the  integral fluctuation relation at stopping times is more powerful than the integral fluctuation relation at fixed times --- in the sense of bounding the likelihood of events of negative entropy production --- is  because with stopping times we can describe fluctuations 
   of entropy production that are not accessible with fixed-time arguments, such as fluctuation of infima of entropy production.     
  
The integral fluctuation relation at stopping times implies also bounds on other fluctuation properties of entropy production, not necessarily related to events of negative entropy production.  For example, we have used the integral fluctuation relation to derive bounds on     splitting probabilities of entropy production and  on the number of times entropy production crosses a certain interval.   For continuous processes these fluctuation properties of entropy production are universal, and we have obtained generic expressions for these fluctuation properties of entropy production.  

Since martingale theory has proven to be very useful to  derive  generic results about the fluctuations of entropy production, the question arises what other physical processes are martingales.     In this context, the exponential of the housekeeping heat has been demonstrated to be a martingale \cite{chetrite2018martingale, chun2018universal}.   The housekeeping heat is an  extension of entropy production to the case of  non-stationary processes.  We expect that all  formulas presented in this paper  extend  in a straightforward manner to the case of  housekeeping heat.  
  Recently also a martingale related to the quenched dynamics of  spin models \cite{ventejou2018progressive} and a  martingale in quantum systems \cite{PhysRevLett.122.220602} have been discovered.        
  
  There exist  universal fluctuation properties that are implied by the martingale property of $e^{-S_{\rm tot}(t)}$, but are not  discussed in this paper.   
For example,  martingale theory implies a   symmetry relation for the distribution of conditional stopping times~\cite{roldan2015decision, saito2015waiting, neri2017statistics, dorpinghaus2018testing}. These symmetry relations are also proved by using the optional stopping theorem, but they cannot  be seen as a straightforward consequence of the integral fluctuation relation at stopping times.

Since the integral fluctuation relation at stopping times  (\ref{stoppingInt}) is a direct consequence of the martingale property $e^{-S_{\rm tot}(t)}$, testing the fluctuation relation (\ref{stoppingInt})  in experiments could serve as a method to   demonstrate that $e^{-S_{\rm tot}(t)}$ is a martingale.  It is not so easy to show in an experiment that a stochastic process is a martingale: it is a herculean task to verify the condition (\ref{eq:equalxo}).    A recent experiment  \cite{singh2017records}  shows that the entropy production of biased transport of single charges in a double electronic dot behaves as a martingale.   The   inequality (\ref{eq:Sinfxx}) for the infima of entropy production was shown to be valid in this experiment.   The integral fluctuation relation at stopping times (\ref{stoppingInt}) provides an interesting alternative to test martingality of $e^{-S_{\rm tot}(t)}$, because the integral fluctuation relation at stopping times is an equality.   Hence, the integral fluctuation relation at stopping times could serve as a proxy for the martingale structure of $e^{-S_{\rm tot}(t)}$ in experiments.

Testing the integral fluctuation relation at stopping times in experimental setting may also be advantageous with respect to testing the standard fluctuation relation at fixed times.      The number of samples required to test the standard fluctuation relation increases exponentially with time, since events of negative entropy are rare.   This makes it difficult to test the conditions of stochastic thermodynamics at large time scales with the standard integral fluctuation relation.   Moreover, at fixed times the distribution of the empirical mean of the   exponentiated negative entropy production  is not known.   The integral fluctuation relation at stopping times does not have these issues, since  negative fluctuations of entropy can be capped at a fixed value $-s_-$, which is independent of time~$t$, and these negative values can be reached at any time.   Moreover, we have  derived an exact universal expression for the distribution of the sample mean  of the exponentiated negative entropy production at stopping times, which can be used to determine the statistical significance of empirical tests of the integral fluctuation relation at stopping times.      The integral fluctuation relation at stopping times is  thus a useful relation  to test  the conditions of stochastic thermodynamics  in a certain experimental setup.

     \appendix

     \section{The exponential of the negative entropy production is a (uniformly integrable) martingale} \label{Append1}
     We  prove that $e^{-S_{\rm tot}(t)}$ is a martingale.   To this aim, we have to verify the two conditions (\ref{eq:integrabl}) and (\ref{eq:equalxo}) for $M(t) = e^{-S_{\rm tot}(t)}$.
We present  two proofs, one for processes in discrete time  using the expression  (\ref{eq:Markov0}) for the entropy production, and a general proof for reversible right-continuous processes $\vec{X}(t)$ using the expression (\ref{eq:SGeneral}) for the entropy production.    

\subsection{Reversible processes in discrete time}
In discrete time 
\begin{eqnarray}
e^{-S_{\rm tot}(t,\omega)}  =  \frac{\tilde{p}(\vec{x}_1, \vec{x}_{2},\ldots, \vec{x}_{t})}{p(\vec{x}_1, \vec{x}_2,\ldots, \vec{x}_t)}
\end{eqnarray}
with $\tilde{p}$ the probability density function associated with the time-reversed dynamics.   We have simplified the notation a bit and used that $\vec{x}(t) = \vec{x}_t$.  

The condition   (\ref{eq:integrabl}) follows from $e^{-S_{\rm tot}(t,\omega)}\geq 0$ and  $E\left[e^{-S_{\rm tot}(t,\omega)}\right]=1$ for all $t\geq 0$.
The martingale condition  (\ref{eq:equalxo}) also holds, 
    \begin{eqnarray}
E\left[e^{-S_{\rm tot}(t,\omega)}|\mathscr{F}_s\right] &=& 
 \int \left(\prod^t_{n=s+1}{\rm d}\vec{x}_n\right)   p(\vec{x}_1, \vec{x}_2,\ldots, \vec{x}_t|\vec{x}_1, \vec{x}_2,\ldots, \vec{x}_s)  \nonumber\\ 
&& \times
   \frac{\tilde{p}(\vec{x}_1, \vec{x}_{2},\ldots, \vec{x}_{t})}{p(\vec{x}_1, \vec{x}_2,\ldots, \vec{x}_t)}  \nonumber\\  
&=&  
 \int \left(\prod^t_{n=s+1}{\rm d}\vec{x}_n\right)   \frac{p(\vec{x}_1, \vec{x}_2,\ldots, \vec{x}_t)}{p(\vec{x}_1, \vec{x}_2,\ldots, \vec{x}_s)}   \frac{\tilde{p}(\vec{x}_1, \vec{x}_{2},\ldots, \vec{x}_{t})}{p(\vec{x}_1, \vec{x}_2,\ldots, \vec{x}_t)}  \nonumber\\ 
&=& \int \left(\prod^t_{n=s+1}{\rm d}\vec{x}_n\right)   \frac{\tilde{p}(\vec{x}_1, \vec{x}_{2},\ldots, \vec{x}_{t})}{p(\vec{x}_1, \vec{x}_2,\ldots, \vec{x}_s)}\nonumber \\ 
&=&  \frac{\tilde{p}(\vec{x}_1, \vec{x}_{2},\ldots, \vec{x}_{s})}{p(\vec{x}_1, \vec{x}_2,\ldots, \vec{x}_s)} \nonumber\\ 
&=& e^{-S_{\rm tot}(s,\omega)}.
\end{eqnarray}    

\subsection{Reversible stationary processes that are right-continuous}
We assume that $\vec{X}(t)$ is rightcontinuous and that the two measures $\mathbb{P}$ and $\mathbb{P}\circ\Theta$ are locally mutually absolutely continuous, such that the entropy production (\ref{eq:SGeneral}) can be defined.  Because of the definition of entropy production, 
\begin{eqnarray}
e^{-S_{\rm tot}(t,\omega)}  =  \frac{\left.{\rm d}(\mathbb{P}\circ\Theta)\right|_{\mathscr{F}_t}}{\left.{\rm d}\mathbb{P}\right|_{\mathscr{F}_t}} .\label{eq:eszc}
\end{eqnarray}
The condition   (\ref{eq:integrabl}) follows from $e^{-S_{\rm tot}(t,\omega)}\geq 0$ and  $E\left[e^{-S_{\rm tot}(t,\omega)}\right]=1$ for all $t\geq 0$.

The martingale condition  (\ref{eq:equalxo}) follows readily from (\ref{eq:eszc}):
 \begin{eqnarray}
E\left[e^{-S_{\rm tot}(t,\omega)}|\mathscr{F}_s\right] = e^{-S_{\rm tot}(s,\omega)},  \quad s \leq t . \label{eq:matingaletobe}
\end{eqnarray}      
The relation (\ref{eq:matingaletobe}) is  a direct consequence of (\ref{eq:eszc}) and the definition of the Radon-Nikodym derivative: The Radon-Nikodym derivative  $e^{-S_{\rm tot}(s,\omega)} = \frac{\left.{\rm d}(\mathbb{P}\circ\Theta)\right|_{\mathscr{F}_s}}{\left.{\rm d}\mathbb{P}\right|_{\mathscr{F}_s}}$  is by definition a $\mathscr{F}_s$-measurable random variable for which 
\begin{eqnarray}
E_{\mathbb{P}}\left[ I_{\Phi}  \frac{\left.{\rm d}(\mathbb{P}\circ\Theta)\right|_{\mathscr{F}_s}}{\left.{\rm d}\mathbb{P}\right|_{\mathscr{F}_s}}  \right] =   E_{\mathbb{P}\circ\Theta}\left[I_{\Phi} \right]
\end{eqnarray}
 for all $\Phi\in\mathscr{F}_s$.      We show that $E\left[e^{-S_{\rm tot}(t,\omega)}|\mathscr{F}_s\right]$ is a random variable with this property, and therefore  (\ref{eq:matingaletobe}) is valid.  Indeed, for $\Phi\in\mathscr{F}_s$ it holds that:
  \begin{eqnarray}
 E_{\mathbb{P}}\left[ I_{\Phi}  E_{\mathbb{P}}\left[e^{-S_{\rm tot}(t,\omega)}|\mathscr{F}_s\right]   \right] &=&  E_{\mathbb{P}}\left[  E_{\mathbb{P}}\left[ I_{\Phi}e^{-S_{\rm tot}(t,\omega)}|\mathscr{F}_s\right]   \right]  \nonumber\\ 
 &=&  E_{\mathbb{P}}\left[  I_{\Phi}e^{-S_{\rm tot}(t,\omega)} \right] \nonumber\\ 
 &=& E_{\mathbb{P}\circ\Theta}\left[I_{\Phi} \right].
 \end{eqnarray}
    Note that the martingale condition (\ref{eq:equalxo})  is consistent with the tower property of conditional expectations: $E\left[E\left[X|\mathscr{F}_t\right]|\mathscr{F}_s\right] = E\left[X|\mathscr{F}_s\right]$.

\subsection{Uniform integrability}
     We prove that the stochastic process $e^{-S_{\rm tot}(t\wedge \tau)}$, with $t\in [0,\infty)$ and $\tau$ a fixed positive number, is uniformly integrable.    The process $e^{-S_{\rm tot}(t\wedge \tau)}$ can be written as 
     \begin{eqnarray}
 e^{-S_{\rm tot}(t\wedge \tau)} =     E\left[e^{-S_{\rm tot}(\tau,\omega)}|\mathscr{F}_t\right]. 
     \end{eqnarray}
     Since a stochastic process $Y(t)$ of the form $Y(t) = E[Z|\mathscr{F}_t]$ is uniformly integrable \cite{liptser2013statistics}, we obtain that  $e^{-S_{\rm tot}(t\wedge \tau)}$ is uniformly integrable.
     
\section{Integral fluctuation relation for entropy production within infinite-time windows } \label{Append2} 
We derive two corollaries of the optional stopping theorem, which is Theorem 3.6   in \cite{liptser2013statistics} and equation  (\ref{eq:Doob1}) in this paper.  
\begin{corollary}\label{corr1}
Let $T$ be a stopping time of  a stationary process $\vec{X}(t)$ and let $S_{\rm tot}(t)$ be the stochastic entropy production of $\vec{X}(t)$ as defined  in (\ref{eq:SGeneral}), with the two measures $\mathbb{P}$ and $\mathbb{P}\circ\Theta$ locally mutually absolutely continuous.   If $t$ is continuous, then $S_{\rm tot}(t)$  is assumed to be rightcontinuous.  
 If the two conditions
   \begin{enumerate}[(i)]
\item $\mathbb{P}(T<\infty) = 1$,
\item $\lim_{t\rightarrow \infty} E\left[ e^{-S_{\rm tot}(t)} I_{T>t}\right] = 0$, 
\end{enumerate}    
are met, then
\begin{eqnarray}
\langle e^{-S_{\rm tot}(T)}  \rangle = 1.
\end{eqnarray}
\end{corollary}
Recall that  $I_{\Phi}(\omega)$ is the indicator function define in (\ref{eq:ind}).  We use a proof analogous to the discrete time proof of Theorem 8.3.5 on page 222 of reference \cite{kannan1979introduction}.

\begin{proof}
     We decompose  $\langle e^{-S_{\rm tot}(T)} \rangle$ into three terms, 
      \begin{eqnarray} 
      \langle e^{-S_{\rm tot}(T)} \rangle =   \langle e^{-S_{\rm tot}(T\wedge t)} \rangle   -  \langle e^{-S_{\rm tot}(t)}I_{T>t} \rangle  +  \langle e^{-S_{\rm tot}(T)} I_{T>t} \rangle. \nonumber
     \end{eqnarray}  
     Since the  right-hand side holds for arbitrary values of $t$ we can take $\lim_{t\rightarrow \infty}$.   Using 
  (\ref{stoppingInt}) and the  condition (ii) we obtain  
     \begin{eqnarray} 
      \langle e^{-S_{\rm tot}(T)} \rangle =  1   +  \lim_{t\rightarrow \infty}\langle e^{-S_{\rm tot}(T)} I_{T>t} \rangle .     \end{eqnarray}  
Because of   condition (i),  it holds that $ \lim_{t\rightarrow \infty} e^{-S_{\rm tot}(T)} I_{T>t} = 0$ in the $\mathbb{P}$-almost sure sense.    Because 
$e^{-S_{\rm tot}(T)} I_{T>t}$ is a nonnegative monotonic decreasing sequence we can apply the monotone convergence theorem, see e.g.~\cite{tao2011introduction},    and we obtain
\begin{eqnarray}
\lim_{t\rightarrow \infty}\langle e^{-S_{\rm tot}(T)} I_{T>t} \rangle =\langle  \lim_{t\rightarrow \infty} e^{-S_{\rm tot}(T)} I_{T>t} \rangle  = 0.
\end{eqnarray}
\end{proof}
     
     \begin{corollary}
Let $T$ be a stopping time of  a stationary process $\vec{X}(t)$ and let $S_{\rm tot}(t)$ be the stochastic entropy production of $\vec{X}(t)$ as defined  in (\ref{eq:SGeneral}), with the two measures $\mathbb{P}$ and $\mathbb{P}\circ\Theta$ locally mutually absolutely continuous.       If $t$ is continuous, then $S_{\rm tot}(t)$  is assumed to be rightcontinuous.    If the two conditions
   \begin{enumerate}[(i)]
\item $\lim_{t\rightarrow \infty}e^{-S_{\rm tot}(t)} = 0$ in the $\mathbb{P}$-almost sure sense, 
\item there exist two positive numbers $s_-$ and $s_+$ such that $S_{\rm tot}(t)\in (-s_-, s_+)$ for all $t\leq T$,
\end{enumerate}    
are met, then
\begin{eqnarray}
\langle e^{-S_{\rm tot}(T)}  \rangle = 1. \label{eq:inegral}
\end{eqnarray}
\end{corollary}
     
     \begin{proof} 
We show that if the  conditions of the present corollary  are met, then also  the conditions of corollary \ref{corr1} are met, and therefore (\ref{eq:inegral}) holds.  

 Because of condition (i), the stopping time $T$ is almost surely finite:  $ 1 =  \lim_{t\rightarrow \infty}\mathbb{P}\left[S_{\rm tot}(t)>s_+\right] \leq  \mathbb{P}\left[T <\infty\right]  $ and $\mathbb{P}\left[T <\infty\right]   \leq 1$, therefore $\mathbb{P}\left[T <\infty\right]   = 1$.
 
 Because of condition (i) and (ii) we obtain $\lim_{t\rightarrow \infty} E\left[ e^{-S_{\rm tot}(t)} I_{T>t}\right] = 0$: $e^{-S_{\rm tot}(t)} I(T>t) \leq e^{s_-}I(T>t)$  is a positive variable that is  bounded from above.  Hence, the dominated convergence theorem applies and $\lim_{t\rightarrow \infty} E\left[ e^{-S_{\rm tot}(t)} I_{T>t}\right] =E\left[\lim_{t\rightarrow \infty}  e^{-S_{\rm tot}(t)} I_{T>t}\right]  = 0$.
     \end{proof}

   \section{Derivation of the bound (\ref{eq:formulaToBe}) on  the statistics of $N_{\times}$} \label{sec:crossinC}
We denote by $M_{\times}(t)$  the number of times entropy production has crossed the interval  $[-\Delta, \Delta]$  in the direction $-\Delta \rightarrow \Delta$ within the time interval $[0,t]$, and therefore $N_{\times} = \lim_{t\rightarrow \infty}M_{\times}(t)$.  

We define two sequences of stopping times $T_n$ and $\tilde{T}_n$, with $n\in [0, N_{\times}]\cap \mathbb{Z}$, namely, 
\begin{eqnarray}
T_n =  {\rm inf}\left\{t: M_{\times}(t)\geq n\right\} 
\end{eqnarray}
and 
\begin{eqnarray}
\tilde{T}_n &=& {\rm inf}\left\{t : t \geq T_n, S_{\rm tot}(t) \notin (-\Delta, s_+) \right\},
\end{eqnarray}  
where $s_+$ is considered to be a very large positive number.

We apply the fluctuation relation (\ref{stoppingInt2}) to the two stopping times $\tilde{T}_n$ and  $T_n$.    The integral fluctuation relation (\ref{stoppingInt2}) implies that 
 \begin{eqnarray}
 E\left[e^{-S_{\rm tot}(\tilde{T}_{n})}|N_{\times} \geq n\right] = E\left[ e^{-S_{\rm tot}(T_n)}|N_{\times} \geq n\right].  \label{eq:eqautionCentral}
 \end{eqnarray} 
 Since $e^{-S_{\rm tot}(T_n)} \leq  e^{-\Delta}$,  the right-hand side of (\ref{eq:eqautionCentral}) is bounded by
  \begin{eqnarray}
E\left[ e^{-S_{\rm tot}(T_n)}|N_{\times} \geq n\right] \leq  e^{-\Delta}. \label{eq:ineq1}
 \end{eqnarray} 
The left-hand side of  (\ref{eq:eqautionCentral}) can be decomposed into two terms,
\begin{eqnarray}
\lefteqn{ E\left[e^{-S_{\rm tot}(\tilde{T}_{n})}|N_{\times} \geq n\right] =  E\left[e^{-S_{\rm tot}(\tilde{T}_{n})}I_{S_{\rm tot}(\tilde{T}_n)\geq s_+ }|N_{\times} \geq n\right]}&&  \nonumber \\ 
&&  +E\left[e^{-S_{\rm tot}(\tilde{T}_{n})}I_{S_{\rm tot}(\tilde{T}_n)\leq  -\Delta }|N_{\times} \geq n\right] ,
 \end{eqnarray}  
 where $I$ is the indicator function defined  in (\ref{eq:ind}).    We take the limit $s_+\rightarrow \infty$ and obtain
\begin{eqnarray}
\lim_{s_+\rightarrow \infty} E\left[e^{-S_{\rm tot}(\tilde{T}_{n})}|N_{\times} \geq n\right] &=& \lim_{s_+\rightarrow \infty}E\left[e^{-S_{\rm tot}(\tilde{T}_{n})}I_{S_{\rm tot}(\tilde{T}_n)\leq  -\Delta}|N_{\times} \geq n\right] \nonumber\\  
&\geq & e^{\Delta}  \lim_{s_+\rightarrow \infty}E\left[I_{S_{\rm tot}(\tilde{T}_n)\leq  -\Delta }|N_{\times} \geq n\right] .\nonumber
 \end{eqnarray} 
Since  \begin{eqnarray}
\mathbb{P}\left[N_{\times}\geq n+1| N_{\times}\geq n\right] = \lim_{s_+\rightarrow \infty}E\left[I_{S_{\rm tot}(\tilde{T}_n)\leq  -\Delta }|N_{\times} \geq n\right]  
 \end{eqnarray} 
we obtain the inequality 
 \begin{eqnarray}
 \lim_{s_+\rightarrow \infty} E\left[e^{-S_{\rm tot}(\tilde{T}_{n})}|N_{\times} \geq n\right] \geq e^{\Delta}\: \mathbb{P}\left[N_{\times}\geq n+1| N_{\times}\geq n\right]t.\label{eq:ineq2}
 \end{eqnarray}
 The relation (\ref{eq:eqautionCentral}) together with the two inequalities (\ref{eq:ineq1}) and (\ref{eq:ineq2})  imply
 \begin{equation}
\mathbb{P}\left[N_{\times}\geq n+1| N_{\times}\geq n\right] \le e^{-2\Delta}\qquad \mathrm{with}\quad n>0.
\end{equation}
This is the formula (\ref{eq:formulaToBe}) which we were meant to prove.

\section*{Acknowledgements}
We thank A.~C.~Barato, R.~Belousov, S.~Bo, R.~Ch\'{e}trite, R.~Fazio, S.~Gupta,  D.~Hartich, C.~Jarzynski, R.~Johal, C.~Maes, G.~Manzano, K.~Netočný,  P.~Pietzonka, M.~Polettini, K.~Proesmans,  S.~Sasa, F.~Severino, and K.~Sekimoto  for  fruitful discussions.  

\section*{References}
\bibliographystyle{ieeetr}
\bibliography{bibliography}

 \end{document}